\documentclass[a4paper,twocolumn,accepted=2024-02-16]{quantumarticle}
\pdfoutput=1
\usepackage{comment}
\usepackage{tikz}
\usepackage{amsmath}
\usepackage{amssymb}  
\usepackage{amsthm}
\usepackage{amsfonts}
\usepackage{graphicx}
\usepackage{bbm}
\usepackage{url}
\usepackage[utf8]{inputenc}
\usepackage[english]{babel}
\usepackage[T1]{fontenc}
\usepackage{braket}
\usepackage{upgreek}
\usepackage{dsfont}
\usepackage{makecell}
\usepackage[numbers,sort&compress]{natbib}

\usepackage{fancyref} %
\usepackage{algorithm}
\usepackage{algpseudocode}

\usepackage{float}
\newfloat{algorithm}{t}{lop}

\usepackage{hyperref}

\theoremstyle{plain}   
\newtheorem{theorem}{Theorem}[section]
\newtheorem{thm}{Theorem}[section]

\newtheorem{lem}{Lemma}[section]

\newtheorem{prop}{Proposition}[section]

\newtheorem{defi}{Definition}[section]
\newtheorem{exmp}{Example}[section]
\theoremstyle{remark}
\newtheorem{rem}{Remark}[section]

\newcommand*{\fancyrefthmlabelprefix}{thm}
\newcommand*{\fancyreflemlabelprefix}{lem}
\newcommand*{\fancyrefcorlabelprefix}{cor}
\newcommand*{\fancyrefdefilabelprefix}{defi}
\frefformat{plain}{\fancyreflemlabelprefix}{lemma\fancyrefdefaultspacing#1}
\Frefformat{plain}{\fancyreflemlabelprefix}{Lemma\fancyrefdefaultspacing#1}
\frefformat{plain}{\fancyrefthmlabelprefix}{theorem\fancyrefdefaultspacing#1}
\Frefformat{plain}{\fancyrefthmlabelprefix}{Theorem\fancyrefdefaultspacing#1}
\frefformat{plain}{\fancyrefcorlabelprefix}{corollary\fancyrefdefaultspacing#1}
\Frefformat{plain}{\fancyrefcorlabelprefix}{Corollary\fancyrefdefaultspacing#1}
\frefformat{plain}{\fancyrefdefilabelprefix}{definition\fancyrefdefaultspacing#1}
\Frefformat{plain}{\fancyrefdefilabelprefix}{Definition\fancyrefdefaultspacing#1}
\newcommand*{\fancyrefalglabelprefix}{alg}
\newcommand*{\frefalgname}{algorithm}
\newcommand*{\Frefalgname}{Algorithm}
\frefformat{plain}{\fancyrefalglabelprefix}{%
  \frefalgname\fancyrefdefaultspacing#1%
}%
\Frefformat{plain}{\fancyrefalglabelprefix}{%
  \Frefalgname\fancyrefdefaultspacing#1%
}%

\newcommand*{\fancyrefapplabelprefix}{app}
\newcommand*{\frefappname}{appendix}
\newcommand*{\Frefappname}{Appendix}
\frefformat{plain}{\fancyrefapplabelprefix}{%
  \frefappname\fancyrefdefaultspacing#1%
}%
\Frefformat{plain}{\fancyrefapplabelprefix}{%
  \Frefappname\fancyrefdefaultspacing#1%
}%

\definecolor{Green}{HTML}{00AD69}  %

\def\beq{\begin{equation}}
\def\eeq{\end{equation}}
\def\bq{\begin{quote}}
\def\eq{\end{quote}}
\def\ben{\begin{enumerate}}
\def\een{\end{enumerate}}
\def\bit{\begin{itemize}}
\def\eit{\end{itemize}}

\def\sa{\operatorname{sa}}

\def\lb{\left(}
\def\rb{\right)}

\def\l|{\left|}
\def\r|{\right|}

\def\i{\text{id}}
\newcommand\C{\mathbbm{C}}

\newcommand\R{\mathbbm{R}}
\newcommand\N{\mathbbm{N}}
\newcommand\M{\mathcal{M}}
\newcommand\D{\mathcal{D}}

\newcommand\cB{\mathcal{B}}

\newcommand{\cL}{\mathcal{L}}

\newcommand{\ketbra}[1]{|#1\rangle\langle#1|}
\newcommand{\tr}[1]{\operatorname{tr}\left[#1\right]}

\newcommand{\one}{I}
\newcommand{\id}{\text{id}}

\newcommand{\scalar}[2]{\langle#1|#2\rangle}

\newcommand{\daniel}[1]{\textcolor{red}{[daniel] #1}}

\newcommand{\cO}{\mathcal{O}}
\newcommand{\tcO}{\tilde{\mathcal{O}}}

\newcommand{\cH}{\mathcal{H}}
\newcommand{\RR}{\mathbb{R}}

\newcommand{\eps}{\epsilon}
\newcommand{\cN}{\mathcal{N}}

\newcommand{\Tr}{\operatorname{tr}}

\begin{document}

\title{Learning quantum many-body systems from a few copies}

\author{Cambyse Rouz\'{e}}
 \email{rouzecambyse@gmail.com}
\affiliation{Department of Mathematics, Technische Universit\"at M\"unchen, 85748 Garching, Germany}

\author{Daniel Stilck Fran\c{c}a}
\email{daniel.stilck\_franca@ens-lyon.fr}
\affiliation{QMATH, Department of Mathematical Sciences, University of Copenhagen, Denmark}
\affiliation{Univ Lyon, ENS Lyon, UCBL, CNRS, Inria, LIP, F-69342, Lyon Cedex 07, France}

\normalsize

\date{11.03.2024}

\begin{abstract}
Estimating physical properties of quantum states from measurements is one of the most fundamental tasks in quantum science. In this work, we identify conditions on states under which it is possible to infer the expectation values of all quasi-local observables of a state from a number of copies that scale polylogarithmically with the system's size and polynomially on the locality of the target observables. We show that this constitutes a provable exponential improvement in the number of copies over state-of-the-art tomography protocols. We achieve our results by combining the maximum entropy method with tools from the emerging fields of classical shadows and quantum optimal transport. The latter allows us to fine-tune the error made in estimating the expectation value of an observable in terms of how local it is and how well we approximate the expectation value of a fixed set of few-body observables. 
We conjecture that our condition holds for all states exhibiting some form of decay of correlations and establish it for several subsets thereof. These include widely studied classes of states such as one-dimensional thermal and high-temperature Gibbs states of local commuting Hamiltonians on arbitrary hypergraphs or outputs of shallow circuits.  Moreover, we show improvements of the maximum entropy method beyond the sample complexity that are of independent interest. These include identifying regimes in which it is possible to perform the postprocessing efficiently as well as novel bounds on the condition number of covariance matrices of many-body states.

\end{abstract}
\maketitle

\section{Introduction}

The subject of quantum tomography has as its goal devising methods for efficiently obtaining a classical description of a quantum system from access to experimental data.  However, all tomographic methods for general quantum states inevitably require resources that scale exponentially in the size of the system~\cite{Wright2016,Haah2017}, be it in terms of the number of samples required or the post-processing needed to perform the task. 

Fortunately, most of the physically relevant quantum systems can be described in terms of a (quasi)-local structure. These range from that of a local interaction Hamiltonian corresponding to a finite temperature Gibbs state to that of a shallow quantum circuit. Hence, locality is a physically motivated requirement that brings the number of parameters describing the system to a tractable number. Effective tomographic procedures should be able to incorporate this information. And, indeed, starting from physically motivated assumptions, many protocols in the literature achieve a good recovery guarantee in trace distance from a number of copies that scales \emph{polynomially} with system size~\cite{Cramer2010,Baumgratz_2013,Torlai2018,PhysRevA.101.032321,2004.07266,Eisert2020}.

Furthermore, in many cases, one is interested in learning only physical properties of the state on which tomography is being performed. These are mostly encoded into the expectation values of quasi-local observables that often only depend on reduced density matrices of subregions of the system. By Helstrom's theorem, obtaining a good recovery guarantee in trace distance is equivalent to demanding that the expectation value of \emph{all bounded observables} are close for the two states, a much larger class of observables than quasi-local ones.

It is, in turn, desirable to design tomographic procedures that can take advantage of the fact that we wish to only approximate quasi-local observables, instead of demanding a recovery in trace distance. And some methods in the literature take advantage of that. For instance, the overlapping tomography or classical shadows methods of~\cite{Huang2020,cotler2020quantum,jena_pauli_2019,crawford_efficient_2020} allow for approximately learning all $k$-local reduced density matrices of an $n$-qubit state with failure probability $\delta$ using $\cO(e^{ck}k\log(n\delta^{-1})\epsilon^{-2})$ copies without imposing any assumptions on the underlying state. This constitutes an exponential improvement in the system size compared to the previously mentioned many-body setting at the expense of an undesirable exponential dependency in the locality of the observables.

In light of the previous discussion, it is natural to ask the guiding question of our work: \emph{is it possible to devise a tomography protocol that has a sample complexity that is logarithmic in system size and polynomial in the locality of the observables we wish to estimate?} 

At first, this might sound like a tall order: as we show in Section~\ref{sec:lowerbound} by importing results of~\cite{Devroye2020}, even if we start from the assumption that the underlying state we wish to learn is a high-temperature product state with $n$ qubits, the number of samples required to obtain an estimate that is $\epsilon$ close in trace distance from the target state scales like $\Omega(n\epsilon^{-2})$. Thus, to obtain a sample complexity that is logarithmic in system size we cannot quantify closeness in trace distance and need to resort to more physically motivated distinguishability measures. Moreover, we show in Section~\ref{sec:lowerbound} that even for product states the classical shadows protocol will fail to produce a good estimate for $k$-local observables if the number of samples is not exponential in $k$. We conclude that protocols like shadow tomography on their own cannot achieve our goal of a sample complexity that is polynomial in the locality of the underlying observables and needs to be combined with other estimation methods in a nontrivial way.

Despite these challenges, we provide an affirmative answer for the guiding question above for a large class of physically motivated states. We achieve this by combining two insights. First, we observe that recently introduced \emph{Wasserstein distances}~\cite{Carlen_2019,Rouz2019,gao_fisher_2020,palma_optimal,kiani2021quantum,giacomo_cambyse} are better suited than the trace distance to estimate by how much the expectation values of physically motivated observables can differ on two states. We introduce these distances and motivate this claim below. But in summary, the Wasserstein distance quantifies how well we can distinguish states through observables whose expectation value does not change much when we apply a unitary acting only on a few qubits.
By focusing on the Wasserstein distance instead of the trace distance we can bypass the $\Omega(n\epsilon^{-2})$ lower bound we mentioned previously. Intuitively, this means that exponentially fewer samples are required to estimate all such local expectation values than arbitrary, global ones.

The second insight is to combine techniques from quantum optimal transport with the well-established maximum entropy method~\cite{jaynes_information_1957} and the classical shadows protocols in a novel way. In particular, we will demonstrate that so-called transportation cost inequalities~\cite{Carlen_2019,Rouz2019,gao_fisher_2020,palma_optimal,kiani2021quantum,giacomo_cambyse} allow us to control how well we approximate the expectation value of $k$-local observables by how well we approximate certain observables that only act on a constant number of qubits. Thus, we only use the shadows protocol to estimate the expectation of many observables that are highly local, the regime in which classical shadows excel, and bypass the exponential scaling of only using shadows for such an estimation task.
This way we obtain a provable exponential improvement over known methods of many-body tomography~\cite{Cramer2010,Baumgratz_2013,2004.07266,Torlai2018,PhysRevA.101.032321,Eisert2020} that focus on the trace distance and recent shadow tomography or overlapping tomography techniques~\cite{Huang2020,cotler2020quantum,jena_pauli_2019,crawford_efficient_2020}, as summarized in Table~\ref{tab:summary}.

\begin{table*}[t]
\centering
\begin{tabular}{|c|c|c|c|}
\hline
Structure                                                                     & Assumptions on State               & Assumptions on Observable          & Samples  \\ \hline
Many-body tomography                                                       & Many-body & none & $\textrm{poly}(n,\epsilon^{-1})$~\cite{2004.07266}   \\ \hline
Classical Shadows & none      & $k$-local                      & $\textrm{poly}(e^{ck},\log(n),\epsilon^{-1})$ \cite{Huang2020}    \\ \hline
This work                                                              & Many-body+transportation     & $k$-local       & $\textrm{poly}(k,\log(n),\epsilon^{-1})$ \\ \hline
\end{tabular}

\caption{Summary of underlying assumptions and sample complexity of other approaches to perform tomography on quantum many-body states.}\label{tab:summary}
\end{table*}

Examples for which we obtain exponential improvements include thermal states of $1$D systems and high-temperature thermal states of commuting Hamiltonians on arbitrary hypergraphs and outputs of shallow circuits. Furthermore, based on results by~\cite{kuwahara2020gaussian,tang_hamiltonian}, we conjecture that our results should hold for any high-temperature Gibbs state, even. More ambitiously, we conjecture that our results can be extended to states exhibiting exponential decay of correlations. This would allow us to extend our findings to classes of states that are not known to be tractable classically, such as ground states of gapped Hamiltonians in higher dimensional lattices~\cite{hastings_spectral_2006}.

The main ingredient to obtain our improvements are so-called transportation cost (TC) inequalities~\cite{Talagrand_1996}. They allow us to bound the difference of expectation values of Lipschitz observables, a concept we will review shortly, on two states by their relative entropy. Such inequalities constitute a powerful tool from the theory of optimal transport~\cite{villani_optimal_2009} and are traditionally used to prove sharp concentration inequalities~\cite[Chapter 3]{raginsky_concentration_2014}. Moreover, they have been recently extended to quantum states~\cite{palma_optimal,Rouz2019,giacomo_cambyse}. By combining such inequalities with the maximum entropy principle, we are able to easily control the relative entropy between the states and, thus, the difference of expectation values of Lipschitz observables.

Our revisit of the maximum entropy principle is further motivated by recent breakthroughs in Hamiltonian learning~\cite{2004.07266,tang_hamiltonian}, shadow tomography~\cite{Huang2020}, the understanding of correlations and computational complexity of quantum Gibbs states~\cite{kliesch_locality_2014,PhysRevLett.124.220601,harrow2020classical,kuwahara2020gaussian,kuwahara_improved_2021} and quantum functional inequalities~\cite{capel2020modified,giacomo_cambyse} that shed new light on this seasoned technique. 

Before we summarize our contributions in more detail, we first define and revise the main concepts required for our results, namely Lipschitz observables, transportation cost inequalities and the maximum entropy principle.
\subsection{Lipschitz observables}

In the classical setting, given a metric $d$ on a sample space $\mathcal{S}$, the regularity of a function $f:\mathcal{S}\to\mathbb{R}$ can be quantified by its \emph{Lipschitz constant}~\cite[Chapter 3]{raginsky_concentration_2014}  
\begin{align}\label{Lipsclass}
\|f\|_{\operatorname{Lip}}=\sup_{x,y\in\mathcal{S}}|f(x)-f(y)|/d(x,y)\,.
\end{align}
For instance, if we consider functions on the $n$-dimensional hypercube $\{-1,1\}^n$ endowed with the Hamming distance, the Lipschitz constant quantifies by how much a function can change per flipped spin. It  should then be clear that physical quantities like average magnetization have a small Lipschitz constant. 
Some recent works~\cite{Rouz2019,palma_optimal} extended this notion to the noncommutative setting and we will focus on the approach of~\cite{palma_optimal} in the main text. This is justified by the fact that it is more intuitive and technically simpler. For the approach followed in~\cite{palma_optimal}, the Lipschitz constant of an observable on $n$ qudits is defined as \footnote{the definition of~\cite{palma_optimal} has a different normalization and does not have the $\sqrt{n}$ term. This normalization will be convenient to treat the constants of~\cite{Rouz2019} and~\cite{palma_optimal} on an equal footing.} 
\begin{align}\label{equ:definitiongiacomo}
    \frac{\|O\|_{\operatorname{Lip},\square}}{\sqrt{n}}:=\max\limits_{1\leq i\leq n}\underset{\substack{\rho,\sigma\in\D_{d^n}\\\operatorname{tr}_i[\rho]=\operatorname{tr}_i[\sigma]}}{\max}\,\tr{O(\rho-\sigma)}\,,
\end{align}
where $\D_{d^n}$ denotes the set of $n$-qudit states. That is, $\|O\|_{\operatorname{Lip},\square}$ quantifies the amount by which the expectation value of $O$ changes for states that are equal when tracing out one site. It is clear that $\|O\|_{\operatorname{Lip}}\leq2\sqrt{n}\|O\|_\infty$ always holds by H\"older's inequality, but it can be the case that $\|O\|_{\operatorname{Lip},\square}\ll \sqrt{n}\|O\|_{\infty}$. For instance, consider for some $k>0$ the $n$-qubit observable 
\begin{align}\label{equ:average_local}
O=n^{-1}\sum_{i=1}^n\otimes_{j=i}^{i+k}Z_j,
\end{align}
where for each site $j$, $Z_j$ denotes the Pauli observable $Z$ acting on site $j$ and we take addition modulo $n$. It is not difficult to see that $\|O\|_{\operatorname{Lip},\square}=2kn^{-\frac{1}{2}}$, while $\|O\|_\infty=1$. We refer to the discussion in Fig.~\ref{fig:observable} for another example.

Moreover, one can show that shallow local circuits or short-time local dynamics satisfying a Lieb-Robinson bound cannot substantially increase the Lipschitz constant of an observable when evolved in the Heisenberg picture. That is, if we have that $\Phi_t^*$ is the quantum channel that describes some local dynamics at time/depth $t$ in the Heisenberg picture and it satisfies a Lieb-Robinson bound, then we have:
\begin{align*}
    \|\Phi_t^*(O)\|_{\operatorname{Lip},\square}=\cO(e^{v t} \|O\|_{\operatorname{Lip},\square})\,,
\end{align*}
where $v$ denotes the Lieb-Robinson velocity. This result is discussed in more detail in Section~\ref{sec:lipsch} of the supplemental material. Thus, averages over local observables and short-time evolutions thereof all belong to the class of observables that have a small Lipschitz constant when compared to generic observables. These facts justify our claim that quasi-local observables are Lipschitz.

Once we are given a Lipschitz constant on observables or a set of quasi-local observables, we can define a Wasserstein-1 distance on states by duality~\cite{Rouz2019,palma_optimal}. The latter quantifies how well we can distinguish two states by their action on regular or local observables and is given by
\begin{align}\label{equ:wasserstein_def}
{W}_{1}(\rho,\sigma):=\sup_{O:\|O\|_{\operatorname{Lip},\square}\leq 1}\tr{O(\rho-\sigma)}.
\end{align}

The definition~\eqref{equ:wasserstein_def} is in direct analogy with the variational definition of the trace distance, given by:
\begin{align}
\|\rho-\sigma\|_{\operatorname{tr}}:=\sup_{O:\|O\|_{\infty}\leq 1}\tr{O(\rho-\sigma)}.
\end{align}
Note, however, that the two quantities have different scalings. To illustrate this point, let us consider the observable in Eq.~\eqref{equ:average_local}. If we measure the distance in trace distance, then we need that $\|\rho-\sigma\|_{\operatorname{tr}}\leq\epsilon$ to ensure that $\sigma$ approximates the expectation of value of $\rho$ up to $\epsilon$ on $O$. On the other hand, as $\|O\|_{\operatorname{Lip},\square}\leq 2kn^{-\frac{1}{2}}$, the bound $W_1(\rho,\sigma)\leq \tfrac{\epsilon\sqrt{n}}{2k}$ is sufficient to guarantee the same approximation. This difference in scaling is at the heart of our results, as we will see now.

\subsection{Transportation cost inequalities}
The paragraphs before motivated the idea that observables with a small Lipschitz constant capture quasilocal observables, and thus, that controlling the Wasserstein distance between two states gives rise to a more physically motivated distance measure than the trace distance. However, it is a priori not clear how to effectively control the Wasserstein distance between states, as it does not admit a closed formula in terms of eigenvalues like the trace distance. 

In this work, we will achieve this by relating Wasserstein distances to the relative entropy between two states, $D(\rho\|\sigma):=\tr{\rho\,(\log(\rho)-\log(\sigma))}$, for $\sigma$ of full-rank. This can be achieved through the notion of a \textit{transportation cost inequality}: an $n$-qudit state $\sigma$ is said to satisfy a transportation cost inequality with parameter $\alpha>0$ if the Wasserstein distance of $\sigma$ to any other state $\rho$ can be controlled by their relative entropy, i.e.
\begin{align}\label{equ:transportation_entropy}
     {W}_{1}(\rho,\sigma)\leq\sqrt{\frac{D(\rho\|\sigma)}{2\alpha}}
\end{align}
holds for all states $\rho\in\D_{d^n}$. 

Such inequalities are particularly powerful whenever the constant $\alpha$ does not depend on the system size $n$ or does so at most inverse polylogarithmically, and can be thought of as a strengthening of Pinsker's inequality.

Transportation cost inequalities are closely related to the notion of Gaussian concentration~\cite{bobkov_exponential_1999,Rouz2019,palma_optimal}, i.e. that Lipschitz functions strongly concentrate around their mean. Establishing analogs of such concentration inequalities for quantum many-body systems has been a fruitful line of research in the last years and they are related to fundamental questions in statistical physics, see e.g.~\cite{brandao_equivalence_2015,Anshu_2016,kuwahara2020gaussian,tasaki_local_2018,kuwahara_eigenstate_2020}. Although we are certain that inequalities like Eq.~\eqref{equ:transportation_entropy} also shed new light on this matter, here we will focus on their application to learning a classical description of a state through maximum entropy methods. We refer to Table~\ref{tab:performance} for a summary of classes of states known to satisfy it, as discussed in more detail below. Unfortunately, for some important classes, the inequalities are only known for the more technically involved variations of the Wasserstein distance, and we refer the reader to the supplemental material, Section~\ref{sec:examplesLipsch} for precise definitions.

Recent works have established transportation cost inequalities with $\alpha$ either constant or logarithmic in system size for several classes of Gibbs states of commuting Hamiltonians~\cite{capel2020modified,palma_optimal,giacomo_cambyse}. In summary, they are known to hold for local Hamiltonians on arbitrary hypergraphs at high enough temperatures or in $1$D. In this work we enlarge the class of examples by showing them for outputs of short-depth circuits in Sec.~\ref{sec:shallowciurcuits}. Note that Eq.~\eqref{equ:transportation_entropy} is trivial for pure states, as then the relative entropy between that state and any other is always $+\infty$. Thus, we first find an appropriate full-rank approximation of the pure state for which the inequality holds, as we will discuss below.

\subsection{Maximum-entropy methods}

Let us now show how transportation cost inequalities can be combined with maximum entropy methods. Such methods start from the assumption that we are given a set of self-adjoint, linearly independent Hermitian observables over an $n$-qudit system,  $E_1,\ldots,E_m\in\M_{d^n}^{\operatorname{sa}}$, with $\|E_i\|_\infty\leq1$, a maximal inverse temperature $\beta>0$ and the promise that the state we wish to learn can be expressed as:
\begin{align}\label{equ:def_state}
\sigma\equiv   \sigma(\lambda)=\frac{\operatorname{exp}\left(-\beta \sum\limits_{i=1}^m\lambda_i E_i\right)}{\mathcal{Z}(\lambda)}\quad ,
\end{align}
where $\lambda\in\R^m$ with sup norm $\|\lambda\|_{\ell_\infty}\leq1$ and 
\begin{align*}
\mathcal{Z}(\lambda)=\tr{\operatorname{exp}\left(-\beta \sum\limits_{i=1}^m\lambda_i E_i\right)}\,
\end{align*}
the partition function.
Denoting by $e(\lambda)$ the vector with components $e_i(\lambda)=\tr{E_i\,\sigma(\lambda)}$, the crux of the maximum entropy method is that $\lambda$ is the unique optimizer of
\begin{align}\label{equ:max-entropy-restricted_main}
        \min_{\substack{\mu\in\R^m\\ \,\|\mu\|_{\ell_\infty}\leq 1}} \log(\mathcal{Z}(\mu))+\beta\sum\limits_{i=1}^m \mu_ie_i(\lambda)\,,
\end{align}
which gives us a convex variational principle to learn the state given $e(\lambda)$.  We refer to Sec.~\ref{sec:max-entropyprin} for a discussion of the maximum entropy principle and its properties.

Typical examples of observables $E_i$ are e.g. all $2-$local Pauli observables corresponding to edges of a given graph. This models the situation in which we are guaranteed that the state is a thermal state of a Hamiltonian with a known locality structure. More generally, for most of the examples discussed here the $E_i$ are given as follows: we start from a hypergraph $G=(V,E)$ and assume that there is a maximum radius $r_0\in\mathbb{N}$ such that, for any hyperedge $A\in E$, there exists a vertex $v\in V$ such that the ball $B(v,r_0) $ centered at $v$ and of radius $r_0$ includes $A$. The $E_i$ are then given by a basis of the traceless matrices on each hyperedge $A$. This definition captures the notion of a local  Hamiltonian w.r.t. the hypergraph.

Our framework also encompasses pure states after making an appropriate approximation. For instance, in this article we will also consider the outputs of shallow quantum circuits and believe our framework extends to unique ground states of gapped Hamiltonians, which are pure. Indeed, although it might not be a-priori clear, we will show below how the outputs of constant depth constant circuits are contained in the class of ground states of gapped, commuting Hamiltonians. And these are well-approximated by Gibbs states at finite temperature. Although this connection between Gibbs states of commuting Hamiltonians and outputs of shallow circuits might not be obvious at first, it makes some intuitive sense that the outputs of a shallow circuit are characterized by the reduced density matrix of the lightcone of each qubit. And the same holds for Gibbs states: one of their defining properties is that they are uniquely determined by the marginals. And this is the property that links these two classes of states.

Let us specify further what it means to learn the output of a shallow circuit. Suppose that $\ket{\psi}=\mathcal{U}\ket{0}^{\otimes n}$ is the output of an unknown shallow circuit $\mathcal{U}$ of depth $L$ with respect to a known interaction graph $(V,E)$ on $n$ vertices. That is,
\begin{align*}
    \mathcal{U}=\prod_{\ell\in [L]}\,\bigotimes_{e\in \mathcal{E}_{\ell}}\,\mathcal{U}_{\ell,e}\,,
\end{align*}
where each $\mathcal{E}_{\ell}\subset E$ is a subset of non-intersecting edges. Thus, in this setting the locality of the circuit is known, but the underlying unitary is not. As we then show in Theorem~\ref{shallowcircuits}, the state $\ket{\psi}$ is $\epsilon^2$ close in Wasserstein distance to the Gibbs state $\sigma$ corresponding to the Hamiltonian with local terms $ \mathcal{U}Z_i\mathcal{U}^\dagger$ at the inverse temperature $\beta=\log(\epsilon^{-1})$. By a simple light-cone argument we can bound the support of each $\mathcal{U}Z_i\mathcal{U}^\dagger$, since we know the underlying structure of the circuit. We then show in Thm.~\ref{shallowcircuits} that it is indeed possible to efficiently learn the outputs of such circuits as long as the support of each time-evolved $Z_i$ is at most logarithmic in system size.

We see from Eq.~\eqref{equ:max-entropy-restricted_main} that the expectation values of the $E_i$ completely characterize the state $\sigma(\lambda)$. But it is possible to obtain a more quantitative version of this statement through the following identity, also observed in~\cite{2004.07266}:
\begin{align}
      &D(\sigma(\mu)\|\sigma(\lambda))+D(\sigma(\lambda)\|\sigma(\mu))\nonumber\\
      &\qquad\qquad  =-\beta\,\scalar{\lambda-\mu}{e(\lambda)-e(\mu)}.\label{equ:rel_entrtopy_Gibbs}
\end{align}
In addition to showing that if $e(\lambda)=e(\mu)$ then $\sigma(\mu)=\sigma(\lambda)$, Eq.~\eqref{equ:rel_entrtopy_Gibbs} implies that by controlling how well the local expectations values of one state approximate another, we can also easily control their relative entropies. In particular, if $m=\cO(n)$ and $\|e(\mu)-e(\lambda)\|_{\ell_1}=\cO(\epsilon n)$ for some $\epsilon>0$, we obtain from an application of H\"older's inequality that:
\begin{align}\label{equ:relativeentropydensity}    D(\sigma(\mu)\|\sigma(\lambda))&\leq\nonumber D(\sigma(\mu)\|\sigma(\lambda))+D(\sigma(\lambda)\|\sigma(\mu))\\
&=\cO(\beta\epsilon n )\,.
\end{align}
We refer to Section \ref{sec:max-entropyprin} for more details. Thus, if we can find a state that approximates the expectation value of each $E_i$ up to $\epsilon$, we are guaranteed to have a $\cO(\beta\epsilon)$ relative entropy density. This observation is vital to ensure that the maximum entropy principle still yields a good estimate of the state even under some statistical noise in the vector of expectation values $e(\lambda)$. Indeed, the variational principle of Eq.~\eqref{equ:max-entropy-restricted_main} would allow us to recover the state exactly if we had access to the exact values of $e(\lambda)$. However, it turns out that solving Eq.~\eqref{equ:max-entropy-restricted_main} with some estimate $\hat{e}(\lambda)$ such that $\|\hat{e}(\lambda)-e(\lambda)\|_{\ell_\infty}\leq \epsilon$ still yields a Gibbs state $\sigma(\mu)$ satisfying Eq.~\eqref{equ:relativeentropydensity}. The maximum entropy problem is a strictly convex optimization problem. Thus, it can be solved efficiently with access to the gradient of the target function. The gradient turns out to be proportional to $e(\lambda)-e(\mu)$, where $\mu$ is the current guess for the optimum. Although we will discuss the details of solving the problem later in Sec.~\ref{sec:max-entropyprin}, in a nutshell, the maximum entropy problem can be solved efficiently if it is possible to efficiently compute expectation values of the observables $E_i$ on the family of Gibbs states under consideration.

\subsection{Combining TC with the maximum entropy principle}
Suppose now that we have that each of the $E_i$ acts on at most $l_0$ qudits. Then, by using e.g. the method of classical shadows, we can estimate the expectation values of all $E_i$ up to $\epsilon$ with failure probability at most $\delta>0$ with $\cO(4^{l_0}\epsilon^{-2}\log(m\delta^{-1}))$ samples. From our discussion above, we see that this is enough to obtain a state $\sigma(\mu)$ satisfying Eq.~\eqref{equ:relativeentropydensity}. Further assuming that we have a TC with some constant $\alpha>0$ for $\sigma(\lambda)$ we conclude that:
\begin{align*}
    \left|\tr{O(\sigma(\lambda)-\sigma(\mu))}\right|&\leq \|O\|_{\operatorname{Lip}}W_1(\sigma(\lambda),\sigma(\mu))\\
    &\leq
    \|O\|_{\operatorname{Lip}}\sqrt{\frac{D(\sigma(\mu)\|\sigma(\lambda))}{2\alpha}}\\
    &=\cO(\sqrt{\beta\epsilon n}\,\|O\|_{\operatorname{Lip}}).
\end{align*}

Finally, recall that for sums of $k$ local operators on a $2$D lattice like in Fig.~\ref{fig:observable}, where we have $k=L^2$, the Lipschitz constant satisfies $\|O\|_{\operatorname{Lip}}=\cO(\sqrt{n})$ and we require a precision of $\cO(\epsilon n/k)$ to obtain a relative error of $\epsilon$. Putting all these elements together, we conclude that by setting $\epsilon=\beta \,\tilde{\epsilon}^2/(k^2)$ for some $\tilde{\epsilon}>0$ we arrive at 
\begin{align}
    \left|\tr{O(\sigma(\lambda)-\sigma(\mu))}\right|=\cO(\tilde{\epsilon}\,k^{-1}n),
\end{align}
which constitutes a relative error for the expectation value. In particular, we see that the sample complexity required to obtain this error was 
\begin{align}\label{equ:numbersamples}
   \cO(4^{l_0}k^4\beta^2\,\tilde{\epsilon}^{-4}\log(m)).
\end{align}
We then obtain:
\begin{theorem}[Learning Gibbs states]\label{thn:learning}
 Let  $\sigma(\lambda)$ be a Gibbs state as defined in Eq.~\eqref{equ:def_state} and such that each $E_i$ acts on at most $l_0$ qudits. Moreover, suppose that $\sigma(\lambda)$ satisfies a $\operatorname{TC}$ inequality with $\alpha$ depending at most inverse logarithmically with system size. Then with probability of success $1-\delta$ we can obtain a state $\sigma(\mu)$ such that for all observables $O\in\M_{d^n}$ 
\begin{align}\label{equ:lipschitzrecovery2}
\left|\tr{O(\sigma(\lambda)-\sigma(\mu))}\right|\leq \cO(\epsilon\sqrt{n}\,\|O\|_{\operatorname{Lip}})\,
\end{align}
from $\cO(4^{l_0}\beta^2\,\operatorname{poly}(\epsilon^{-1},\log(m\delta^{-1})))$ samples of $\sigma(\lambda)$. Moreover, if it is possible to compute the expectation values of the $E_i$ on $\sigma(\tau)$ for $\|\tau\|_{\ell_\infty}\leq 1$, then the postprocessing can also be in polynomial time.
\end{theorem}
We once again stress that the recovery guarantee in Eq.~\eqref{equ:lipschitzrecovery2} suffices to give good relative approximations for the expectation value of quasilocal observables. Furthermore, if we did not resort to the Wasserstein distance but rather the trace distance, as in known results for the tomography of many-body states~\cite{Cramer2010,Baumgratz_2013,Lanyon2017,PhysRevA.101.032321,2004.07266}, the sample complexity would be exponentially worse, as we prove in Sec.~\ref{sec:lowerbound}. More precisely, any algorithm that estimates Gibbs states on a lattice at inverse temperatures $\beta=\Omega(n^{-\frac{1}{2}})$ up to trace distance $\epsilon$ requires $\Omega(n\epsilon^{-2})$ samples. Thus, even for states whose inverse temperature goes to $0$ as the system size increases, a focus on the trace distance instead of the Wasserstein distance implies an exponentially worse sample complexity. 

\begin{table*}[ht!]

    \centering
    
    \begin{tabular}{|l|l|l|l|}
    \hline
    Structure                                                                     & Samples               & Postprocessing             \\ \hline
    Lightcone $l_0$ circuit                                                       & $\epsilon^{-4}2^{12l_0}$ & $\operatorname{exp}(3n)$    \\ \hline
    \begin{tabular}[c]{@{}l@{}}Commuting Gibbs\\ for $\beta<\beta_c$\end{tabular} & $\epsilon^{-2}$       & $\textrm{poly}(n)$                       \\ \hline
    $1$D commuting                                                                & $\min\{\epsilon^{-2}e^{\mathcal{O}(\beta)},\epsilon^{-4}\}$       & $\textrm{poly}(n)$            \\ \hline
    \end{tabular}
    \caption{Performance of the algorithm under various assumptions of the underlying state to obtain a state that is $\epsilon\sqrt{n}$ close in Wasserstein distance. All the estimates are up to $\textrm{polylog}(n)$ factors. We refer to Sec.~\ref{sec:TCIs} for proofs of the TC used and~\ref{sec:Wasslearning} for how to combine them with maximum entropy methods. In Section~\ref{sec:strong_convexity} we explain how to obtain the sample complexity by combining Thm.~\ref{thm:comcomplexitygeneral} with strong convexity bounds. For the postprocessing we refer to section~\ref{sec:regimesefficient}. The case of shallow circuits is discussed in more detail in Sec.~\ref{sec:shallowciurcuits}. By lightcone $l_0$ we mean the size of the largest lightcone of each qubit in the circuit.}\label{tab:performance}
    \end{table*}
    
Theorem~\ref{thn:learning} also provides an exponential improvement over shadow techniques in the locality of the observables, as we argue in Sec.~\ref{sec:lowerbound}. However, unlike our methods, shadow techniques do not need to make any assumptions on the underlying states. Thus, we see that Theorem~\ref{thn:learning} opens up the possibility of highly efficient characterization of quantum states and provably exponentially better sample complexities when compared with recovery in trace distance.

We also remark that it is possible to improve the scaling in accuracy in Eq.~\eqref{equ:numbersamples} from $\epsilon^{-4}$ to the expected $\epsilon^{-2}$. To do that, it is important to bound the condition number of the Hessian of the log-partition function, as we explain in the methods.

\begin{figure}[ht!]
\centering
\includegraphics[width=1\columnwidth]{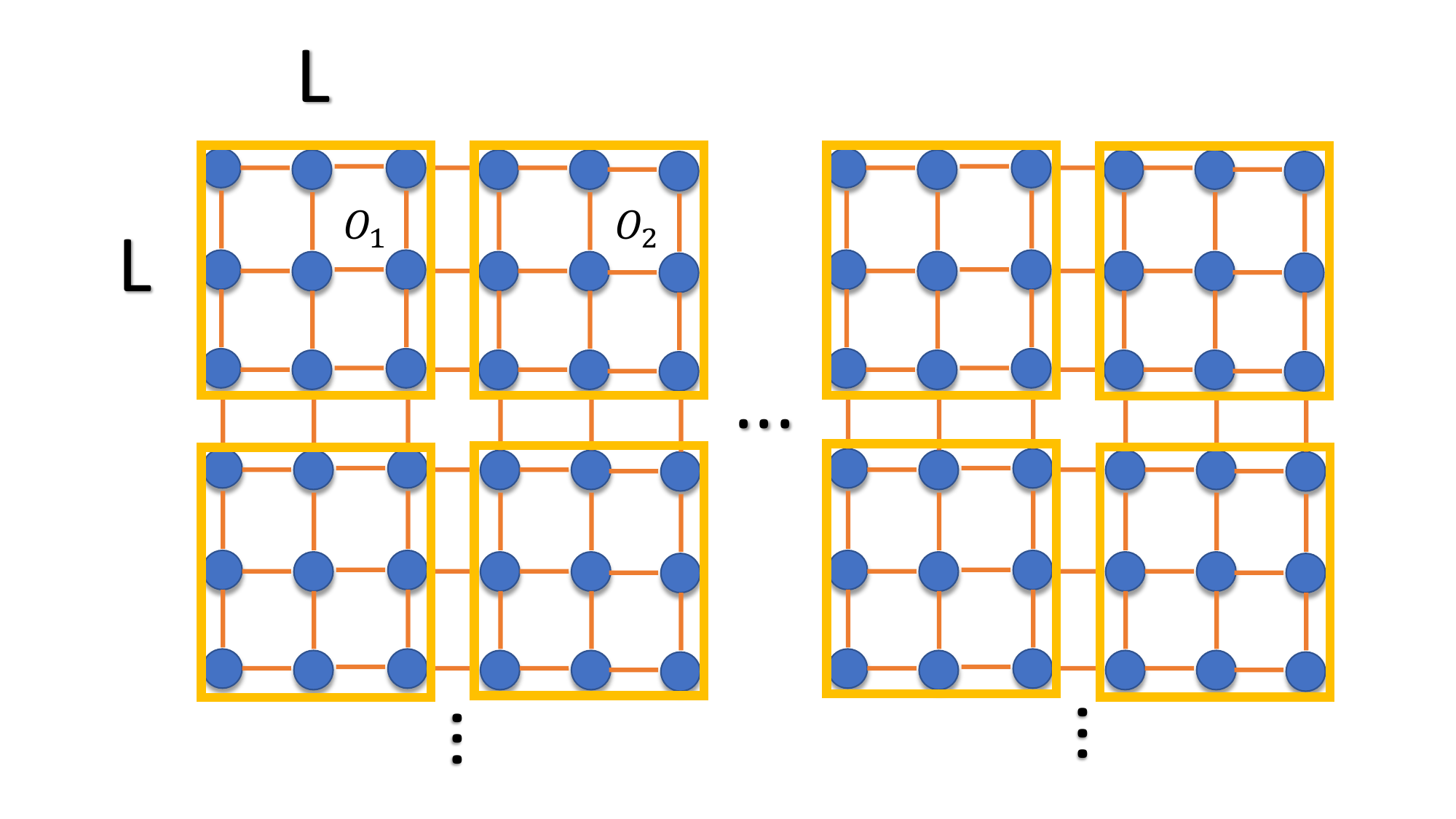}
\caption{example of observable $O=\sum_iO_i$ for $2D$ lattice system of size $n$. Each $O_i$ is supported on a $L\times L$ square ($L=3$ in the figure). We have $\|O\|_{\operatorname{Lip}}=\cO(\sqrt{n})$ and $\|O\|=n/L^2$. Thus our methods require $\textrm{poly}(L,\log(n),\epsilon^{-1})$ samples to estimate the expectation value of all such observables. Shadow-like methods require $\textrm{poly}(e^{cL^2},\log(n),\epsilon^{-1})$ samples, an exponentially worse dependency in $L$. Even for moderate values of $L$, say $L=5$, this can lead to $10^7$ factor savings sample complexity and gives an exponential speedup for $L=\textrm{poly}(\log(n))$. Other many-body methods have a $\textrm{poly}(L,n,\epsilon^{-1})$~\cite{Cramer2010,Baumgratz_2013,Lanyon2017,PhysRevA.101.032321,2004.07266} scaling, which in turn is exponentially worse in the system size.}
\label{fig:observable}
\end{figure}

\subsection{Numerical results}
We will now compare the performance of our method to the classical shadow protocol~\cite{Huang2020} to estimate the average of a local observable on a Gibbs state. To ensure that we can still generate samples for a high number of qubits, we will consider the following family of commuting Gibbs states in $1$-D:
\begin{align}\label{equ:form_Hamiltonian}
&H(\lambda)=\\
&~~-\sum\limits_{k=0}^{n/2-1}S^{2k}(\lambda_kX_0X_1+\lambda_{n+k}X_0X_1Y_2Y_3)S^{2k},\nonumber
\end{align}
where $S$ is the shift operator, $\lambda\in B_{\ell_\infty}(0,1)$ and we assumed $n$ is even. We will then estimate the expectation value of the observable:
\begin{align}\label{equ:observable_numerics}
O=\sum\limits_{k=0}^{n/2-1}S^{2k}X_0X_1Y_2Y_3X_4X_5Y_6Y_7S^{2k}.
\end{align}
The results for one particular choice of Gibbs state in this class are shown in Fig.~\ref{fig:shadows_vs_gibbs}. It shows that even for observables of moderate locality like the one in Eq.~\eqref{equ:observable_numerics}, shadows are outperformed by maximum entropy methods by orders of magnitude. Also note that the quality of our estimates decays like $\sim1/\sqrt{s}$, where $s$ is the number of samples, showing how the quality of the recovery is essentially independent of the system's size.

We also remark that for obtaining these results, we obtained the expectation values of the $X_iX_{i+1}$ terms by measuring in the $X$ basis on each qubit and of the $XXYY$ terms by measuring in a sequence of $XXYY$ bases followed by the same basis shifted by $2$.

\begin{figure*}[!ht]
\centering
\includegraphics[width=0.6\textwidth]{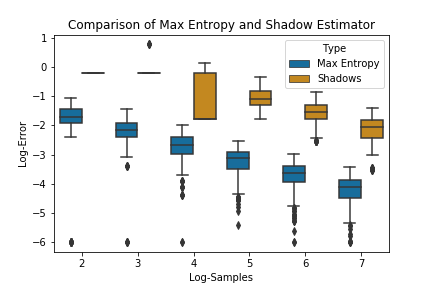}
\caption{Performance on a Gibbs state from the family of Eq.~\eqref{equ:form_Hamiltonian} and the 8-local observable in Eq.~\eqref{equ:observable_numerics}. 
We have set the number of qubits to $100$, $\beta=1.1$ and the $\lambda_i$ uniformly at random between $0.5$ and $0.9$.
The $x$-axis denotes the logarithm of the number of samples in base $10$ and the $y$ the error in absolute value to the true value. We ran each protocol $300$ times on the same Gibbs state to see how the estimate varied. We see that even with $10^{4}$ samples the shadows method still has errors of the order $10^{0}$ in the $75$ percentile, whereas the maximum entropy already yields good estimates when the number of samples is of order $10^2$, showcasing that maximum entropy methods outperform classical shadows by orders of magnitudes for observables of moderate locality like those in Eq.~\eqref{equ:observable_numerics}.}
\label{fig:shadows_vs_gibbs}
\end{figure*}

\section{Conclusion} In this article we have demonstrated that ideas from quantum optimal transport yield provable exponential improvements in the sample complexity required to learn a quantum state when compared to state-of-the-art methods. More precisely, we showcased how the interplay between maximum entropy methods and the Wasserstein distance, which is mediated by a transportation cost inequality, allows for fine-tuning the complexity of observables whose expectation value we wish to estimate and the number of samples required for that. Through our techniques, we essentially settled most questions related to how efficiently it is possible to learn a commuting Gibbs state and significantly advanced our understanding of general Gibbs states. With the impressive growth in the size of quantum devices available in the lab over the last years, we believe that the polylogarithmic in system size sample complexities obtained here will come in handy to calibrate and characterize systems containing thousands or millions of qubits.

We believe that the framework and philosophy we began to develop here will find applications in other areas of quantum information theory. Indeed, although a bound in trace distance is the golden standard and one of the most widely accepted and used measures of distance between quantum states in quantum information and computation, we argued that in many physically relevant settings, demanding a trace distance bound might be too strong. More importantly, replacing the trace distance by a Wasserstein distance bound can lead to exponential complexity gains, as in this article. Thus, we believe that this approach is likely to lead to substantial gains and improvements also in other areas like quantum machine learning, process tomography or in quantum many-body problems.

Some of the outstanding open questions raised by this article are establishing that a suitable notion of exponential decay of correlations in general implies a transportation cost inequality and showing that TC holds for a larger class of systems. We believe that our framework should also extend to ground states of gapped Hamiltonians in 1D, however, such a statement would still require us to refine our bounds. The results presented here also make us conjecture that any high-temperature Gibbs state satisfies a TC inequality, even for long-range interactions, which would make our techniques applicable to essentially all physically relevant models at high temperatures.

That being said, to the best of our knowledge, all states that are known to satisfy a TCI can also be simulated efficiently classically. However, there is a-priori no reason to believe that a TCI necessarily implies that the underlying states can be simulated classically and we believe that this is more of a reflection of the fact that the study of TCI is still incipient. Furthermore, although we explained how to do the tomography of the Gbibs state with the maximum entropy principle, other Hamiltonian learning procedures also apply, as all we need is to upper-bound the relative entropy through Eq.~\eqref{equ:rel_entrtopy_Gibbs}. Thus, breakthroughs in Hamiltonian learning would then also apply to the setting in this paper. In particular, through efficient Hamiltonian learning methods, it might be possible to use the results derived here to verify if a quantum device correctly implemented a shallow circuit.

Moreover, it would also be interesting to investigate other applications of Gaussian concentration in many-body physics~\cite{brandao_equivalence_2015,Anshu_2016,kuwahara2020gaussian,tasaki_local_2018,kuwahara_eigenstate_2020} from the angle of transportation cost inequalities.

\section{Methods}
\subsection{Summary of the maximum entropy procedure and contributions}
Now that we have discussed how our results yield better sample complexities for some classes of states, we discuss the maximum entropy algorithm in more detail and comment on how our results equip it with better performance guarantees.

As the maximum entropy principle in Eq.~\eqref{equ:max-entropy} corresponds to solving a convex optimization problem, it should come as no surprise that promises on the strong convexity of the underlying functions being optimized can be leveraged to give improved performance guarantees~\cite{boyd2004convex}.
For the specific case of the maximum entropy principle, strong convexity guarantees translate to bounds of the form
\begin{align}\label{equ:conditionnummbermain}
  L I\leq \nabla^2 \log(\mathcal{Z}(\mu))\leq U I  
\end{align}
for constants $L,U>0$ and all $\mu\in B_{\ell_\infty}(0,1)$.
We refer to Sec.~\ref{sec:max-entropyprin} of the supplemental material for a thorough discussion. We note that in~\cite{2004.07266} the authors show such results in a more general setting, although with $U,L$ polynomial in $n$, which is not sufficient for our purposes.  For us it will be important to ensure that the condition number of the log-partition function is at most polylogarithmic in system size (i.e. $L^{-1}U=\tilde{\cO}(1)$). 

If we define the function $f:\mu\mapsto\log(\mathcal{Z}(\mu))+\beta\sum\limits_{i=1}^m \mu_ie_i(\lambda)$, then $\nabla f=\beta(e(\lambda)-e(\mu))$. It then follows from standard properties of strongly convex functions that:
\begin{align*}
    \|\lambda-\mu\|_{\ell_2}=L^{-1}\beta\|e(\lambda)-e(\mu)\|_{\ell_2}
\end{align*}
That is, whenever the expectation values are close, the underlying parameters must be close as well. In this case, we have from $\|e(\lambda)-e(\mu)\|_{\ell_2}=\cO(\epsilon\sqrt{m})$ that:
\begin{align}\label{equ:error_estimate_strong_convexity}
    &-\beta\,\scalar{\lambda-\mu}{e(\lambda)-e(\mu)}\\&\leq\beta\|\lambda-\mu\|_{\ell_2}\|e(\lambda)-e(\mu)\|_{\ell_2}=\cO(L^{-1}\beta^2\epsilon^2m).\nonumber
\end{align}
As we will see in the proposition below, $L^{-1}=\cO(e^{\beta}\beta^{-2})$ for commuting Hamiltonians, which gives a quadratic improvement in $\epsilon$ in Eq.~\eqref{equ:error_estimate_strong_convexity} and yields the expected $\epsilon^{-2}$ scaling for the sample complexity in terms of the accuracy. 
\begin{prop}[Strengthened strong convexity constant for commuting Hamiltonians]\label{prop:lowerbound_main}
For each $\mu\in B_{\ell_\infty}(0,1)$, let $\sigma(\mu)$ be a Gibbs state at inverse temperature $\beta$ corresponding to the commuting Hamiltonian $H(\mu)=\sum_j\,\mu_j\,E_j$ on the hypergraph $G=(V,E)$, where $\tr{E_iE_j}=0$ for all $i,j$ and each local operator $E_j$ is traceless on its support. Then for $\beta$ such that the states $\sigma(\mu)$ satisfy exponential decay of correlations, the Hessian of the log-partition function is bounded by
\begin{align}\label{equ:strong_convexity_commuting}
    \Omega(
        \beta^2\,e^{-c\beta })I\leq \nabla^2\,\log\,\mathcal{Z}(\mu)\leq  \cO(\beta^2)I.
\end{align}
for some constant $c>0$.
\end{prop}
After the completion of the first version of this work, in~\cite[Corollary 4.4]{tang_hamiltonian} Tang et al proved strong convexity bounds for high-temperature, (not necessarily geometrically) local Hamiltonians. More precisely, for $\beta=\cO(k^{-8})$, where $k$ is the maximal number of qudits each term acts on, they show that $L^{-1}\leq 2\beta^{-2}$. Although the result in Eq.~\eqref{equ:strong_convexity_commuting} has the advantage of giving an estimate at any temperature, we see that also for noncommuting Gibbs states strong convexity holds at high enough temperatures.

The flowchart in Figure~\ref{fig:flowchart} gives the general scheme behind the maximum entropy method. Besides the exponential improvements in sample complexity laid out in Table~\ref{tab:performance}, we also provide structural improvements which we elaborate on while also explaining the general scheme:

\paragraph{\textbf{Input:}} The input consists of $m$ linearly independent operators $E_i$ of operator norm $1$, some upper bound $\beta>0$, precision parameter $\epsilon>0$ and step size $\eta$. Moreover, we are given the promise that the state of interest satisfies~\eqref{equ:def_state}. Although we will be mostly concerned with the case in which the observables are local, we show the convergence of the algorithm in general in Sec.~\ref{sec:max-entropyprin}. The step size should be picked as $\eta=\cO(U^{-1})$ with $U$ satisfying \eqref{equ:conditionnummbermain}, as explained in Sec.~\ref{sec:convergence_guarantees}.

\paragraph{\textbf{Require}:} We assume that we have access to copies of $\sigma(\lambda)$ and that we can perform measurements to estimate the expectation values of the observables $E_i$ up to precision $\epsilon>0$. For most of the examples considered here, this will mostly require implementing simple, few-qudit measurements.

\paragraph{\textbf{Output}:} The output is in the form of a vector of parameters $\mu$ of a Gibbs state $\sigma(\mu)$ as in Eq.~\eqref{equ:def_state}. Note that unlike~\cite{2004.07266}, our goal is not to estimate the vector of parameters $\lambda$, but rather obtain an approximation of the state satisfying $\sigma(\lambda)\simeq \sigma(\mu)$. Here we will focus on quantifying the output's quality in relative entropy. More precisely, the output of the algorithm is guaranteed to satisfy $D(\sigma(\mu)\|\sigma(\lambda))=\cO(\epsilon n)$.

\paragraph{\textbf{Step 1}:} In this step, we estimate the expectation values of each observable $E_i$ on the state $\sigma(\lambda)$ up to an error $\epsilon$. The resources to be optimized here are the number of samples of $\sigma(\lambda)$ we require and the complexity to implement the measurements. Using shadow tomography or Pauli grouping methods~\cite{Huang2020,cotler2020quantum,bonet2020nearly} we can do so requiring $\cO(4^{r_0}\epsilon^{-2}\operatorname{polylog}(m))$ samples and Pauli or $1-$qubit measurements, where $r_0$ is maximum number of qubits the $E_i$ act on. This is discussed in more detail in Sec.~\ref{sec:Wasslearning}.

\paragraph{\textbf{Step 2}:} The maximum entropy problem in Eq.~\eqref{equ:max-entropy-restricted_main} can be solved with gradient descent, as it corresponds to a strictly convex optimization problem~\cite{boyd2004convex}. At this step we simply initialize the algorithm to start at the maximally mixed state.

\paragraph{\textbf{Step 3}:} It turns out that the gradient of the maximal entropy problem target function at $\mu_t$ is proportional to $e(\mu_t)-e(\lambda)$. Thus, to implement an iteration of gradient descent, it is necessary to compute $e(\mu_t)$, as we assumed we obtained an approximation of $e(\lambda)$ in Step 1. Moreover, it is imperative to ensure that the algorithm also converges with access to approximations to $e(\mu)$ and $e(\lambda)$. This is because most algorithms to compute $e(\mu_t)$ only provide approximate values~\cite{hastings_quantum_2007,kliesch_locality_2014,PhysRevLett.124.220601,harrow2020classical} . In addition, they usually have a poor scaling in the accuracy~\cite{PhysRevLett.124.220601}, making it necessary to show that the process converges with rough approximations to $e(\mu_t)$ and $e(\lambda)$. Here we show that it is indeed possible to perform this step with only approximate computations of expectation values. This allows us to identify classes of states for which the postprocessing can be done efficiently. These results are discussed in more detail in Sec.~\ref{sec:convergenceapprox}.

\paragraph{\textbf{Convergence loop}:} Now that we have seen how to compute one iteration of gradient descent, the next natural question is how many iterations are required to reach the stopping criterium. As this is a strongly convex problem, the convergence speed depends on the eigenvalues of the Hessian of the function being optimized~\cite[Section 9.1.2]{boyd2004convex}. For max-entropy, this corresponds to bounding the eigenvalues of a generalized covariance matrix. In~\cite{2004.07266} the authors already showed such bounds for local Hamiltonians implying the convergence of the algorithm in a number of steps depending polynomially in $m$ and logarithmically on the tolerance $\epsilon$ for a fixed $\beta$.
Here we improve their bound in several directions. First, we show that the algorithm converges after a polynomial in $m$ number of iterations for arbitrary $E_i$, albeit with a polynomial dependence on the error, as discussed in Sec.~\ref{sec:convergenceapprox}. 
We then specialize to certain classes of states to obtain various improvements. 
For high-temperature, commuting Hamiltonians we provide a complete picture and show that the condition number of the Hessian is constant in Prop.~\ref{prop:lowerbound_main}. This implies that gradient descent converges in a number of iterations that scales logarithmically in system size and error.

\begin{figure*}[!ht]
\centering
\includegraphics[width=0.8\textwidth]{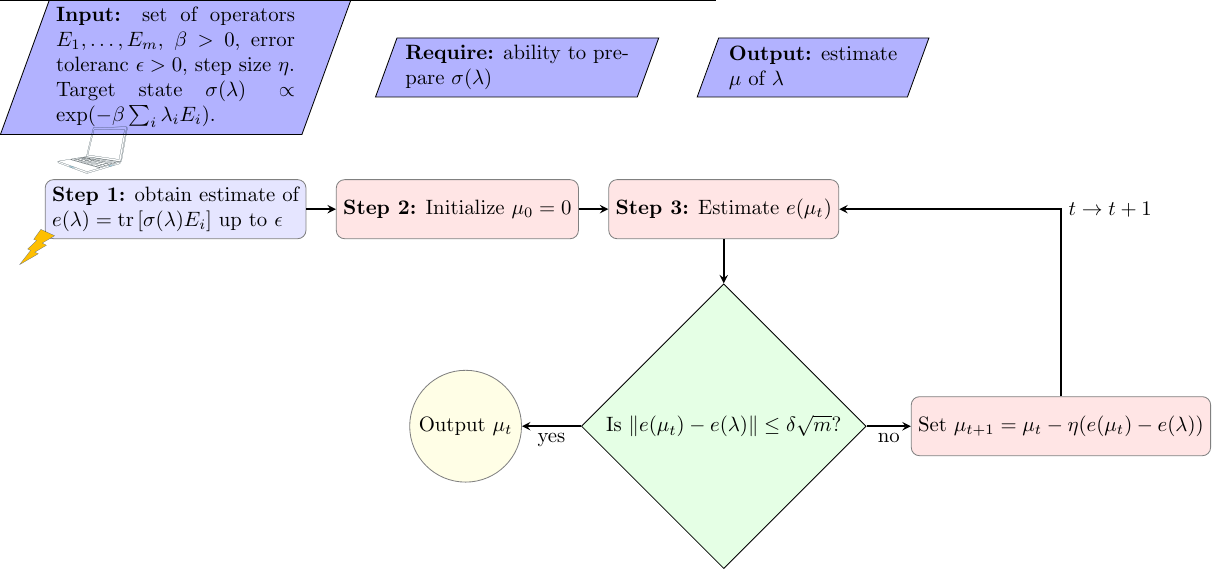}
\caption{Flowchart for general maximum entropy algorithms. }
\label{fig:flowchart}
\end{figure*}
\paragraph{\textbf{Stopping condition and recovery guarantees}:} the stopping condition, 
\begin{align*}
    \|e(\lambda)-e(\mu_t)\|_{\ell_2}\leq\epsilon \sqrt{m}\,,
\end{align*}
can be immediately converted to a relative entropy bound between the target state and the current iterate by the identity~\eqref{equ:rel_entrtopy_Gibbs}. This justifies its choice as a stopping criterion. 

Since we already discussed the sample complexity of the maximum entropy method, let us now discuss some of its computational aspects. There are two quantities that govern the complexity of the algorithm: how many iterations we need to perform until we converge and how expensive each iteration is.

As the maximum entropy problem is strongly convex, one can show that $\cO(UL^{-1}\log(n\epsilon^{-1}))$ iterations suffice to converge. Here again $U,L$ are bounds on the Hessian as in Eq.~\eqref{equ:conditionnummbermain}. 
Nevertheless, we also show how to bypass requiring such bounds in Sec.~\ref{sec:convergence_guarantees} and obtain that the maximum entropy algorithm converges after $\cO(m\epsilon^{-2})$ iterations without any locality assumptions on the $E_i$  or strong convexity guarantees. That is, the number of iterations is at most polynomial in $m$.

Let us now discuss the cost of implementing each iteration of the algorithm on a classical computer. This boils down to estimating $e(\mu_t)$ for the current iterate, which can be achieved in various ways. In the worst case, it is possible to just compute the matrix exponential and the expectation values directly, which yields a complexity of $\cO(d^{3n}m)$. However, for many of the classes considered here it is possible to do this computation in polynomial time. For instance, in~\cite{PhysRevLett.124.220601} the authors show that for high-temperature Gibbs states it is possible to approximate the partition function efficiently. Thus, for the case of high-temperature Gibbs states, not necessarily commuting ones, we can do the postprocessing efficiently. It is also worth mentioning tensor network techniques to estimate $e(\mu_t)$. As we only require computing the expectation value of local observables, recent  works show that it is possible to find tensor network states of constant bond dimension that approximate all expectation values of a given locality well~\cite{dalzell_locally_2019,alhambra_locally_2021,huang_locally_2021}. From such a representation it is then possible to compute $e(\mu_t)$ efficiently in the $1$D case by contracting the underlying tensor network. Unfortunately, however, in higher dimensions the contraction still takes exponential time. Table~\ref{tab:performance} provides a summary of the complexity of the postprocessing for various classes.

It is also worth considering the complexity of the postprocessing with access to a quantum computer, especially for commuting Hamiltonians. As all high-temperature Gibbs states satisfy exponential decay of correlations, the results of~\cite{brandao_finite_2019} imply that high-temperature Gibbs states can be prepared with a circuit of depth logarithmic in system size. Thus, by using the same method we used to estimate $e(\lambda)$ we can also estimate $e(\mu_t)$ by using the copies provided by the quantum computer. The complexity of the postprocessing for shadows is linear in system size. Thus, with access to a quantum computer we can perform the  post-processing for each iteration in  time $\tilde{O}(m\epsilon^{-2})$. As in this case we showed that the number of iterations is $\tcO(1)$, we conclude that we can perform the postprocessing in time $\tilde{\cO}(m\epsilon^{-2})$. That is, for this class of systems our results give an arguably complete picture regarding the postprocessing, as it can be performed in a time comparable to writing down the vector of parameters, up to polylogarithmic factors. Furthermore, given that the underlying Gibbs states are known to satisfy TC and Prop.~\ref{prop:lowerbound_main} gives essentially optimal bounds on the covariance matrices, we believe that the present work essentially settles the question of how efficiently we can learn such Hamiltonians and corresponding Gibbs states.  We discuss this in more detail in Sec.~\ref{sec:completepicturecommuting}.

Finally, an example of our bounds is illustrated in Fig.~\ref{fig:pinsker}, where we show that the number of samples required to estimate a local observable to relative precision is essentially system-size independent.
\begin{figure}[ht!]
\includegraphics[width=1\columnwidth]{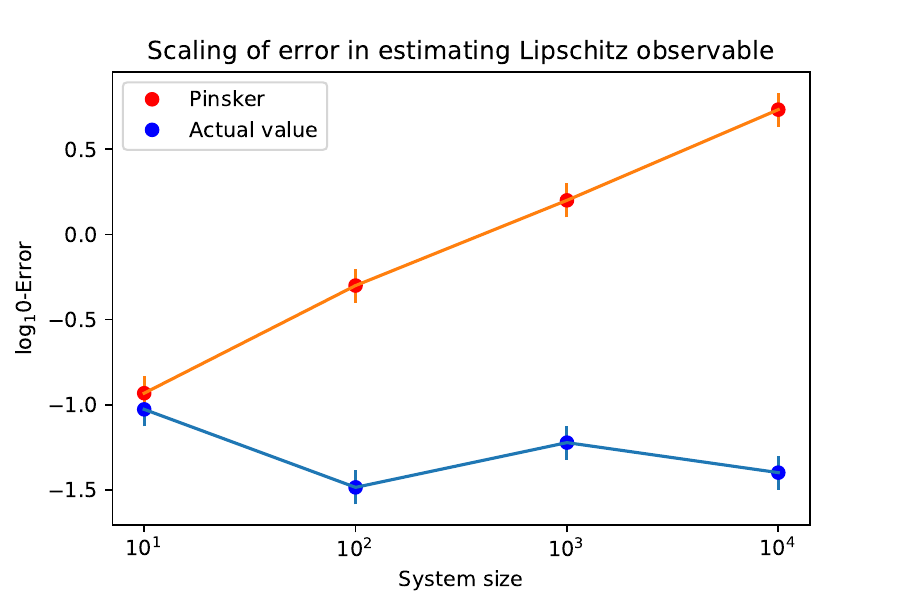}
\caption{Error in estimating a Lipschitz observable after performing the maximum entropy reconstruction method. The underlying state is a classical $1$D-Gibbs state with randomly chosen nearest-neighbor interactions and inverse temperature $\beta=1$. We estimated all the $Z_iZ_{i+1}$ expectation values from the original state based from $10^{3}$ samples of the original state. We then computed the upper bound on the trace distance predicted by Eq.~\eqref{equ:rel_entrtopy_Gibbs} and Pinsker's inequality and compared it to the  actual of discrepancy for a Lipschitz observable on the reconstructed and actual state. The Lipschitz observable was chosen as $\sum_in^{-1}UZ_iZ_{i+2}U^\dagger$, where we picked $U$ as a depth $3$ quantum circuit. We observe that the error incurred is essentially independent of system size, and we get good predictions even when the number of samples is smaller than it.}
\label{fig:pinsker}
\end{figure}

\section{Acknowledgements}

DSF was supported by the VILLUM FONDEN via the QMATH Centre of Excellence under Grant No. 10059 and from the European Research Council (grant agreement
no. 81876).
The research of CR has been supported by project QTraj (ANR-20-CE40-0024-01) of the French National Research Agency (ANR) and by  a  Junior  Researcher  START  Fellowship  from  the MCQST. 
DSF and CR are grateful to Richard Kueng, Fernando Brand\~ao and Giacomo De Palma for interesting discussions.

\bibliographystyle{plainnat}
\bibliography{tomography_doi}

\onecolumngrid
\clearpage
\appendix

\section*{Supplemental Material}
This is the supplemental material to "Learning many-body states from very few copies". We will start in Sec.~\ref{sec:max-entropyprin} with a review of the basic properties of the maximum entropy principle to learn quantum states. This is followed by a discussion of Lipschitz constants, Wasserstein distances and transportation cost inequalities in Sec.~\ref{Transportcost}. After that, in Sec.~\ref{sec:Wasslearning} we more explicitly discuss the interplay between the maximum entropy method and transportation cost inequalities. We then briefly discuss scenarios in which the postprocessing required for the maximum entropy method can be performed efficiently in Sec.~\ref{sec:regimesefficient}. In Sec.~\ref{sec:completepicturecommuting} we discuss a class of examples where we show that all technical results required to obtain the strongest guarantees of our work hold, that is, Gibbs states of commuting Hamiltonians at high enough temperature and 1D commuting Hamiltonians. Finally, in Sec.~\ref{sec:lowerbound} we discuss lower bounds on the sample complexity of both shadow protocols and many-body algorithms that focus on a recovery in trace distance.

We start by setting some basic notations. Throughout this article, we denote by $\mathcal{M}_k$ the algebra of $k\times k$ matrices on $\mathbb{C}^k$, whereas $\mathcal{M}_k^{\sa}$ denotes the subspace of self-adjoint matrices. The set of quantum states over $\mathbb{C}^k$ is denoted by $\D_{k}$. Typically, $k$ will be taken as $d^n$ for $n$ qudit systems. The trace on $\mathcal{M}_k$ is denoted by $\operatorname{tr}$. Given two quantum states $\rho,\sigma$, we denote by $S(\rho)=-\tr{\rho\log(\rho)}$ the von Neumann entropy of $\rho$, and by $D\lb\rho\|\sigma\rb$ the relative entropy between $\rho$ and $\sigma$, i.e. $D(\rho\|\sigma)=\tr{\,\rho\,(\log(\rho)-\log(\sigma))}$ whenever the support of $\rho$ is contained in that of $\sigma$ and $+\infty$ otherwise. The trace distance is denoted by $\|\rho-\sigma\|_{tr}:=\tr{|\rho-\sigma|}$ and the operator norm of an observable by $\|O\|_{\infty}$. Scalar products are denoted by $\langle\cdot|\cdot\rangle$. Moreover, we denote the $\ell_p$ norm of vectors by $\|\cdot\|_{\ell_p}$, and for $x\in\R^m$ and $r\in\R$, $B_{\ell_p}(x,r)$ denotes the ball of radius $r$ in $\ell_p$ norm around $x$. The identity matrix is denoted by $I$. The adjoint of an operator $A$ is denoted by $A^\dagger$ and that of a channel $\Phi$ with respect to the trace inner product by $\Phi^*$. For a hypergraph $G=(V,E)$  we will denote the distance between subsets of vertices induced by the hypergraph by $\operatorname{dist}$.

\section{Maximum entropy principle for quantum Gibbs states}\label{sec:max-entropyprin}
One of the main aspects of this work concerns the effectiveness of the maximum entropy method for the tomography of quantum Gibbs states in various settings and regimes. Thus, we start by recalling some basic properties of the maximum entropy method. Our starting assumption is that the target state is well-described by a quantum Gibbs state with respect to a known set of operators and that we are given an upper bound on the inverse temperature:
\begin{defi}[Gibbs state with respect to observables]
Given a set of observables $\mathcal{E}=\{E_i\}_{i=1}^m$, $E_1,\ldots,E_m\in\M^{\sa}_{d^n}$ being linearly independent with $\|E_i\|_\infty\le 1$, we call a state $\sigma\in\D_{d^n}$ a Gibbs state at inverse temperature $\beta>0$ if there exists a vector $\lambda\in\R^m$ with $\|\lambda\|_{\ell_\infty}\leq1$ such that:
\begin{align}\label{equ:gibbsstate}
\sigma=\operatorname{exp}\left(-\beta \sum\limits_{i=1}^m\lambda_i E_i\right)/\mathcal{Z}(\lambda),\qquad \text{where}\qquad \mathcal{Z}(\lambda)=\tr{ \operatorname{exp}\left(-\beta\sum\limits_{i=1}^m\lambda_i E_i\right)}
\end{align}
 denotes the partition function. In what follows, we will denote $\sigma$ by $\sigma(\lambda)$ and $\sum_i\lambda_iE_i=H(\lambda)$, where the dependence of $\sigma(\lambda)$ on $\beta$ is implicitly assumed.
\end{defi}
We are mostly interested in the regime where $m\ll d^{n}$. Then the above condition can be interpreted as imposing that the matrix $\log(\sigma)$ is sparse with respect to a known basis $\mathcal{E}$. A canonical example for such states are Gibbs states of local Hamiltonians on a lattice, for which $m=\mathcal{O}(n)$ and the observables $E_i$ are taken as tensor products of Pauli matrices acting on neighboring sites. But we could also consider a basis consisting of quasi-local operators or some subspace of Pauli strings.

Next, we review some basic facts about quantum Gibbs states. One of their main properties is that they satisfy a well-known maximum entropy principle~\cite{jaynes_information_1957}. 
This allows us to simultaneously show that the expectation values of the observables $\mathcal{E}$ completely characterize the state $\sigma(\lambda)$ and further provides us with a variational principle to learn a description from which we can infer an approximation of other expectation values. Let us start with the standard formulation of the maximum entropy principle:
\begin{prop}[Maximum entropy principle]\label{MaxEnt}
Let $\sigma(\lambda)\in\D_{d^n}$ be a quantum Gibbs state~\eqref{equ:gibbsstate} with respect to the basis $\mathcal{E}$ at inverse temperature $\beta$ and introduce $e_i(\lambda):=\tr{\sigma(\lambda)E_i}$ for $i=1,\ldots,m$. Then $\sigma(\lambda)$ is the unique optimizer of the maximum entropy problem:
\begin{align}\label{equ:max-entropy}
    \underset{\rho \in \D_{d^n}}{\operatorname{maximize}}& \quad S(\rho)\\
    \mathrm{subject~ to} & \quad \tr{E_i \rho}=e_i(\lambda) \quad \text{for all $i=1,\ldots,m$.}\quad\nonumber
\end{align}
Moreover,
$\sigma(\lambda)$ optimizes:
\begin{align}\label{equ:max-entropy-restricted}
        \min_{\substack{\mu\in\R^m\\ \,\|\mu\|_{\ell_\infty}\leq 1}} \log(\mathcal{Z}(\mu))+\beta\sum\limits_{i=1}^m\mu_ie_i(\lambda)\,.
        \end{align}
\end{prop}
\begin{proof}
The proof is quite standard, but  we include it for completeness. Note that for any state $\rho\not=\sigma$ that is a feasible point of Eq.~\eqref{equ:max-entropy} we have that:
\begin{align*}
    S(\sigma(\lambda))-S(\rho)&=D(\rho\|\sigma(\lambda))+\tr{(\rho-\sigma(\lambda))\log(\sigma(\lambda))}\\
    &=D(\rho\|\sigma(\lambda))-\beta\sum\limits_{i=1}^m\lambda_i\tr{E_i\lb \rho-\sigma(\lambda)\rb}\\
    &=D(\rho\|\sigma(\lambda))>0\,,
\end{align*}
where we have used the fact that $\tr{E_i\lb \rho-\sigma(\lambda)\rb}=0$ for all feasible points and that the relative entropy between two different states is strictly positive. This shows that $\sigma(\lambda)$ is the unique solution of~\eqref{equ:max-entropy}. Eq.~\eqref{equ:max-entropy-restricted} is nothing but the dual program of Eq.~\eqref{equ:max-entropy}.
\end{proof}

 Eq.~\eqref{equ:max-entropy} above gives a variational principle to find a quantum Gibbs state corresponding to certain expectation values. As it is well-known, one can use gradient descent to solve the problem in Eq.~\eqref{equ:max-entropy-restricted}, as it is a strongly convex problem. Various recent works have discussed learning of Gibbs states~\cite{swingle_reconstructing_2014,2009.09000,2004.07266} and it is certainly not a new idea to do so through maximum entropy methods. Nevertheless, we will discuss how to perform the postprocessing in more detail, as some recent results allow us to give this algorithm stronger performance guarantees. Finally, it should be said that although we draw inspiration from~\cite{2004.07266}, our main goal will be to learn a set of parameters $\mu\in\R^m$ such that the Gibbs states $\sigma(\mu)$ and  $\sigma(\lambda)$ are approximately the same on sufficiently regular observables while optimizing the sample complexity. This is in contrast to the goal of~\cite{2004.07266}, which was to learn the vector of parameters $\lambda$. Learning $\lambda$ corresponds to a stronger requirement, in the sense that if the vectors of parameters are  close, then the underlying states are also close, as made precise in the following Prop.~\ref{prop:recoveryguarantee}.

One of the facts that we are going to often exploit is that it is possible to efficiently estimate the relative entropy between two Gibbs states $\sigma(\lambda)$ and $\sigma(\mu)$ given the parameters $\lambda,\mu$ and the expectation values of observables in $\mathcal{E}$. This also yields an efficiently computable bound on the trace distance.
Indeed, as observed in~\cite{2004.07266}, we have that:
\begin{prop}\label{prop:recoveryguarantee}
Let $\sigma(\mu),\sigma(\lambda)\in\D_{d^n}$ be Gibbs states with respect to a set of observables $\mathcal{E}$ at inverse temperature $\beta$. Denote $e(\lambda)=(\tr{\sigma(\lambda)E_i})_i\in\R^m$.
Then
\begin{align}\label{equ:rel_ent_recovery}
\|\sigma(\mu)-\sigma(\lambda)\|_{tr}^2\leq D(\sigma(\mu)\|\sigma(\lambda))+D(\sigma(\lambda)\|\sigma(\mu))=-\beta\,\scalar{\lambda-\mu}{e(\lambda)-e(\mu)}.
\end{align}
\end{prop}
\begin{proof}
The equality in Eq.~\eqref{equ:rel_ent_recovery} follows from a simple manipulation. Indeed:
\begin{align*}
   D(\sigma(\mu)\|\sigma(\lambda))+D(\sigma(\lambda)\|\sigma(\mu))=\beta\sum\limits_{i=1}^m({\lambda_i-\mu_i})\tr{[\sigma(\mu)-\sigma(\lambda)]E_i}.
\end{align*}
The bound on the trace distance then follows by applying Pinkser's inequality.
\end{proof}
The statement of Proposition~\ref{prop:recoveryguarantee} allows us to obtain quantitative estimates on how well a given Gibbs state approximates another one in terms of the expectation values of known observables. In particular, a simple application of H\"older's inequality shows that if two Gibbs states are such that $|\tr{E_i[\sigma(\mu)-\sigma(\lambda)]}|\leq \epsilon$, then the sum of their relative entropies is at most 
\begin{align}\label{thebound}
  \beta\, |\langle\lambda-\mu|e(\lambda)-e(\mu)\rangle|\le \beta\,\|\lambda-\mu\|_{\ell_\infty}\,m\eps\le 2 m \eps\,\beta\,,
\end{align}
where the outer bound arises from our assumption that $\|\lambda\|_{\ell_\infty},\,\|\mu\|_{\ell_\infty}\le 1$. Moreover, it is straightforward to relate the difference of the target function in Eq.~\eqref{equ:max-entropy-restricted} evaluated at two vectors to the difference of relative entropies between the target state and their corresponding Gibbs states:
\begin{lem}\label{propdiffrelent}
Let $\sigma(\lambda)\in\D_{d^n}$ be a Gibbs state with respect to a set of observables $\mathcal{E}$ at inverse temperature $\beta>0$ and define, for any $\mu\in\mathbb{R}^m$,
\begin{align}\label{ffunction}
         f(\mu):=\log(\mathcal{Z}(\mu))+\beta\sum\limits_{i=1}^me_i(\lambda)\,\mu_i\,.
\end{align}
Then for any other two vectors $\mu,\xi\in \mathbb{R}^m$ with $\|\mu\|_{\ell_\infty},\|\xi\|_{\ell_\infty}\le 1$:
\begin{align*}
    f(\mu)-f(\xi)=D(\sigma(\lambda)\|\sigma(\mu))-D(\sigma(\lambda)\|\sigma(\xi)).
\end{align*}
\end{lem}
\begin{proof}
The proof follows from straightforward manipulations. 
\end{proof}
Thus, we see that a decrease of the target function $f$ when solving the max entropy problem is directly connected to the decrease of the relative entropy between the target state and the current iterate. We will later use this to show the convergence of gradient descent for solving the max entropy problem with arbitrary  $\mathcal{E}$.
However, before that we discuss how the convergence of the state $\sigma(\mu)$ to $\sigma(\lambda)$ is related to the convergence of the parameters $\mu$ to $\lambda$.
\subsection{Strong convexity and convergence guarantees}\label{sec:convergence_guarantees}
The maximum entropy problem~\eqref{equ:max-entropy} being a convex problem, it should come as no surprise that properties of the Hessian of the function being optimized are vital to understanding its complexity and stability. For the maximum entropy problem, the Hessian at a point is given by a generalized covariance matrix corresponding to the underlying Gibbs state. As the results of~\cite{2004.07266} showcase, the eigenvalues of such covariance matrices govern both the stability of Eq.~\eqref{equ:max-entropy-restricted} with respect to $\mu$ and the convergence of gradient descent to solve it. To see why, we recall some basic notions of optimization of convex functions and refer to~\cite{boyd2004convex} for an overview. 
\begin{defi}[Strong convexity]
Let $C\subset\R^m$ be a convex set. A twice differentiable function $f:C\to\R$ is called strongly convex with parameters $U,L>0$ if we have for all $x\in C$ that:
\begin{align*}
UI\geq\nabla^2 f(x)\geq LI.
\end{align*}
\end{defi}
The optimization of strongly convex functions is well-understood. Indeed, we have:
\begin{prop}\label{theogradient}
Let $C\subset\R^m$ be a convex set and $f:C\to \mathbb{R}$ be strongly convex with parameters $L,U$ as in the definition above. Then, for all $\epsilon>0$, the optimal value $\alpha:=\min_{x\in C}f(x)$ is achieved up to error $\epsilon$ by the gradient descent algorithm initiated at $x^0\in C$ with step size $U^{-1}$ after at most $S$ steps for 
\begin{align}\label{equ:gradientdescent_steps}
 S\leq\frac{U}{L} \log\left(\frac{f(x^0)-\alpha}{\epsilon}\right).
\end{align}
Moreover, the gradient norm satisfies $\|\nabla f(x^k)\|_{\ell_2}^2\le \delta$ after at most $S_\nabla$ steps with
 \begin{align*}
     S_\nabla \le \frac{U}{L}\log\left(\frac{2L(f(x^0)-\alpha)}{\delta}\right)\,.
 \end{align*}
 Finally, we have for all $\mu,\lambda\in C$ that:
\begin{align}\label{equ:gradient-parameters}
\|\mu-\lambda\|_{\ell_2}\leq L^{-1}\|\nabla f(\mu)-\nabla f(\lambda)\|_{\ell_2}\,.
\end{align}
\end{prop}
\begin{proof}
These are all standard results that can be found e.g. in \cite[Section 9]{boyd2004convex}.
\end{proof}

To see the relevance of these results for the maximum entropy problem, we recall the following Lemma:
\begin{lem}\label{lem:derivatives}
Let $C=B_{\ell_{\infty}}(0,1)\subset\R^m$, let $\sigma(\lambda)\in\D_{d^n}$ be a Gibbs state with respect to a set of operators $\mathcal{E}$ at inverse temperature $\beta$ and define $f:C\to \RR$ as in Eq.~\eqref{ffunction}. Then:
\begin{align}\label{firstderive}
\lb\nabla f(\mu)\rb_i=\beta\tr{(\sigma({\lambda})-\sigma({\mu}))E_i}
\end{align}
and 
\begin{align}\label{Hessian}
&\lb\nabla^2 f(\mu)\rb_{ij}=\frac{\beta^2}{2}  \tr{\{E_j,\Phi_{H(\mu)}(E_i)\}\sigma(\mu)}  -\beta^2\,e_i(\mu)e_j(\mu)\,,
\end{align}
with 
\begin{align*}
\Phi_{H(\mu)}(E_i)=\int\limits_{-\infty}^{+\infty}\nu_\beta(t)e^{-iH(\mu)t}E_ie^{iH(\mu)t}dt
\end{align*}
where $\nu_\beta(t)$ is a probability density function whose Fourier transform is given by:
\begin{align*}
\hat{\nu}_\beta(\omega)=\frac{2\tanh\left(\frac{\beta\omega}{2}\right)}{\beta\omega}\, .
\end{align*}
\end{lem}
\begin{proof}
The quantum belief propagation theorem~\cite{Hastings_2007} states that:
\begin{align*}
\frac{\partial}{\partial \lambda_i} e^{-\beta H(\lambda)}=-\frac{\beta}{2} \{e^{-\beta H(\lambda)},\Phi_{H(\lambda)}(E_i)\}\,.
\end{align*}
The  claim then follows from a simple computation.
\end{proof}
Thus, we see that in order to compute the gradient of the target function $f$ for the maximum entropy problem, we simply need to compute the expectation values of observables $\mathcal{E}$ on the current state and on the target state. Moreover, the Hessian is given by a generalized covariance matrix of the quantum Gibbs state.
That this should indeed be interpreted as a covariance matrix is most easily seen by considering commuting Hamiltonians. Then indeed we have:
\begin{align*}
\lb\nabla^2 f(\mu)\rb_{ij}=\beta^2 \left[\tr{\sigma(\mu)E_iE_j}-e_i(\mu)e_j(\mu)\right]\,.
\end{align*}
For any Gibbs state it holds that:
\begin{prop}\label{upperconstant}
For all  $\mu\in B_{\ell_{\infty}}(0,1)\subset\R^m$, inverse temperature $\beta>0$ and set of operators $\mathcal{E}$ of cardinality $m$, we have:
\begin{align*}
    \nabla^2 f(\mu)\leq 2\beta^2 m\,I\,.
\end{align*}
\end{prop}
\begin{proof}
Note that 
\begin{align*}
    \left|\tr{E_i\sigma(\mu)e^{iH(\mu)t}E_je^{-iH(\mu)t}}\right|\leq 1\,,
\end{align*}
by H\"older's inequality, the submultiplicativity and unitary invariance of the operator norm and the fact that $\|E_i\|_\infty\leq1$.
Similarly, we have that $|e_i(\mu)|,|e_j(\mu)|\leq1$. Thus, by Lemma~\ref{lem:derivatives}, $|\lb\nabla^2 f(\mu)\rb_{ij}|\leq 2\beta^2$. As $\nabla^2 f(\mu)$ is an $m\times m$ matrix, it follows from Gershgorin's circle theorem that $\nabla^2 f(\mu)\leq2 \beta^2 m I$.
\end{proof}
The proof above also showcases how exponential decay of correlations can be used to sharpen estimates on the maximal eigenvalue of $\nabla^2 f$, since in that case $\lb\nabla^2 f(\mu)\rb_{ij}$ will have exponentially decaying entries. We discuss this in more detail when we focus on many-body states, for which we also consider the more challenging question of lower bounds.

\subsection{Convergence with approximate gradient computation and expectation values} \label{sec:convergenceapprox}
Proposition~\ref{theogradient} already establishes the convergence of gradient descent whenever we can compute the gradient exactly and have a bound on $L$. Moreover, we see from Lemma~\ref{lem:derivatives} that, in order to compute the gradient of the function $f$ above, it suffices to estimate local expectation values. Moreover, it is a standard result that gradient descent is strictly decreasing for strictly convex problems~\cite{boyd2004convex}. 

However, in many settings it is only possible or desirable to compute the expectation values of quantum Gibbs states approximately. Moreover, the expectation values of the target state are only known up to statistical fluctuations.
It is then not too difficult to see that gradient descent still offers the convergence guarantees if we only approximately compute the gradient. We state the exact convergence guarantees and precision requirement for completeness.

\begin{thm}[Computational complexity and convergence guarantees]\label{thm:comcomplexitygeneral}
Let $\sigma(\lambda)\in\D_{d^n}$ be a quantum Gibbs state at inverse temperature $\beta$ with respect to a set of operators $\mathcal{E}$ and $C_{\mathcal{E}}$ be the computational cost of computing $e'(\mu)$ satisfying
\begin{align*}
\|e'(\mu)-e(\mu)\|_{\ell_2}\leq \delta_\mu
\end{align*}
 for $\mu\in B_{\ell_\infty}(0,1)$ and $\delta_\mu>0$. Moreover, assume that we are given an estimate $e'(\lambda)$ of $e(\lambda)$ satisfying 
\begin{align}\label{equ:good_estimate}
\|e(\lambda)-e'(\lambda)\|_\infty\leq\epsilon
\end{align} 
and that the partition function is strongly convex with parameters $U,L$.
Then gradient descent starting at $\mu=0$ with step size $\frac{1}{cU}$ and input data $e'(\lambda)$ converges to a state $\sigma(\mu_*)$ satisfying: 
\begin{align*}
  \|\sigma(\lambda)-\sigma(\mu_*) \|_{tr}^2 & \leq D(\sigma(\lambda)\|\sigma(\mu_*))+D(\sigma(\mu_*)\|\sigma(\lambda))\\
  &= \cO(\beta\delta_\mu\min\{\sqrt{m},\beta L^{-1}\delta_\mu\}+\beta\epsilon\min\{1,L^{-1}\beta\epsilon\}m). 
\end{align*}
in
\begin{align*}
\cO\left(\min\left\{UC_{\mathcal{E}}\beta^{-2} n\log(d)\delta_\mu^{-2},\frac{UC_{\mathcal{E}}}{L}\log(n\epsilon^{-1})\right\}\right)
\end{align*}
time.
\end{thm}
We will prove this theorem at the end of this section, as before we will need some auxiliary statements. But the reader familiar with basic concepts from convex optimization should feel comfortable to skip them.

\begin{prop}[Convergence of gradient descent with constant relative precision]\label{prop:approx_grad}
Let $\sigma(\lambda)\in\D_{d^n}$ be a quantum Gibbs state at inverse temperature $\beta$ with respect to a set of operators $\mathcal{E}$, and for a Gibbs state $\sigma(\mu)$ let $z(\mu)\in\R^m$ be a vector such that 
\begin{align}\label{equ:small_relative_error}
    \|z(\mu)-\beta(e(\lambda)-e(\mu))\|_{\ell_2}\leq \,\frac{\beta}{4c}\|e(\mu)-e(\lambda)\|_{\ell_{\ell_2}}\,
\end{align}
for some $c>10$.
Then we have that:
\begin{align}\label{equ:convergenceentropy}
    D(\sigma(\lambda)\|\sigma(\mu-\tfrac{z(\mu)}{c U}))-D(\sigma(\lambda)\|\sigma(\mu))\leq -\frac{9\beta^2\|e(\mu)-e(\lambda)\|_{\ell_2}^2}{10\,c\,U}\,,
\end{align}
where $U$ is a uniform bound on the operator norm of the Hessian of the function $f$ defined in Eq.~(\ref{ffunction}).
\end{prop}
\begin{proof}
From a Taylor expansion and strong convexity we have for any two points $\mu,\xi$ that:
\begin{align*}
    f(\xi)\leq f(\mu)+\scalar{\nabla f(\mu)}{(\xi-\mu)}+\frac{U}{2}\|\xi-\mu\|_{\ell_2}^2.
\end{align*}
Note that $\nabla f(\mu)=\beta\lb e(\lambda)-e(\mu)\rb$ by Eq.~\eqref{firstderive}.
Setting $\xi=\mu-\frac{z}{cU}=\mu+\frac{1}{cU}\lb-\nabla f(\mu)+\nabla f(\mu)-z\rb$ we obtain:
\begin{align*}
   &f\lb\mu-\frac{z}{cU}\rb\\
   &~~~\leq f(\mu)-\frac{1}{cU}\|\nabla f\|_{\ell_2}^2+\frac{1}{cU}\scalar{\nabla f(\mu)}{\nabla f(\mu)-z}+\frac{1}{2c^2U}\|-\nabla f(\mu)+\nabla f(\mu)-z\|_{\ell_2}^2\\
   &~~~\leq f(\mu)-\frac{1}{cU}\|\nabla f\|_{\ell_2}^2+\frac{1}{cU}\|\nabla f(\mu)\|_{\ell_2}\,\|\nabla f(\mu)-z\|_{\ell_2}+\frac{1}{2c^2U}\lb\|\nabla f(\mu)\|_{\ell_2}+\|\nabla f(\mu)-z\|_{\ell_2}\rb^2,
\end{align*}
where in the last step we used the Cauchy-Schwarz inequality. By our assumption in Eq.~\eqref{equ:small_relative_error} for $z\equiv z(\mu)$ we have:
\begin{align*}
    f(\mu)&-\frac{1}{cU}\|\nabla f\|_{\ell_2}^2+\frac{1}{cU}\|\nabla f(\mu)\|_{\ell_2}\|\nabla f(\mu)-z\|_{\ell_2}+\frac{1}{2c^2U}\lb\|\nabla f(\mu)\|_{\ell_2}+\|\nabla f(\mu)-z\|_{\ell_2}\rb^2\\
    &\leq f(\mu)-\frac{1}{cU}\|\nabla f\|_{\ell_2}^2+\frac{1}{4c^2U}\|\nabla f(\mu)\|_{\ell_2}^2+\frac{(1+(4c)^{-1})^2}{2c^2U}\|\nabla f(\mu)\|_{\ell_2}^2\\
    &=f(\mu)-\frac{1}{cU}\lb1-\frac{1}{4c}-\frac{(1+(4c)^{-1})^2}{2c}\rb\|\nabla f(\mu)\|_{\ell_2}^2
\end{align*}
and it can be readily checked that $\frac{1}{4c}+\frac{(1+(4c)^{-1})^2}{2c}\leq\frac{1}{10}$ for $c\geq 10$. To conclude the proof, note that by Lemma \ref{propdiffrelent}:
\begin{align*}
    D(\sigma(\lambda)\|\sigma(\mu-\tfrac{z}{c U}))-D(\sigma(\lambda)\|\sigma(\mu))=f(\mu-\tfrac{z}{c U}))-f(\mu)\,,
\end{align*}
and insert $\|\nabla f(\mu)\|_{\ell_2}^2=\beta^2\|e(\mu)-e(\lambda)\|_{\ell_2}^2$.
\end{proof}

Thus, we see that we make constant progress for the gradient descent algorithm if we only compute the derivative up to constant relative precision. We show now how to pick our stopping criterium based on approximate computations of the gradient which ensure convergence in polynomial time.
\begin{prop}\label{prop:convergence_grad}
Let $\sigma(\lambda)\in\D_{d^n}$ be a quantum Gibbs state at inverse temperature $\beta$ with respect to a set of operators $\mathcal{E}$. Suppose that at each time step $t$ of gradient descent we compute an estimate $e'(\mu_t)$ of  $e(\mu_t)$ that satisfies 
\begin{align*}
    \|e'(\mu_t)-e(\mu_t)\|_{\ell_{\ell_2}}\leq \delta_\mu\,,
\end{align*}
 and set the stopping criterion to be:
\begin{align*}
    \|e(\lambda)-e'(\mu_*)\|_{\ell_2}< (4c+1)\delta_\mu\,.
\end{align*}
for some constant $c>10$. Then gradient descent starting at $\mu=0$ with update rule $\mu_{t+1}:=\mu_t-\frac{\beta (e(\lambda)-e'(\mu_t))}{cU}$
will converge to a state $\sigma(\mu_*)$ satisfying:
\begin{align*}
  \|\sigma(\lambda)-\sigma(\mu_*) \|_{tr}^2\leq D(\sigma(\lambda)\|\sigma(\mu_*))+D(\sigma(\mu_*)\|\sigma(\lambda))\leq 2(4c+1)\beta\delta_\mu \sqrt{m}
\end{align*}
after at most $\cO(U\beta^{-2} n\log(d)\delta_\mu^{-2})$ iterations. 
\end{prop}
\begin{proof}
First, we show that the relative precision bound required for Proposition~\ref{prop:approx_grad} holds under these assumptions. By our choice of the stopping criterium, at each time step we have the property that, while we did not stop,
\begin{align*}
    \|e(\lambda)-e(\mu_t)\|_{\ell_2}&=\|e(\mu)-e'(\mu_t)+e'(\mu_t)-e(\lambda)\|_{\ell_2}\\
   & \geq \|e'(\mu_t)-e(\lambda)\|_{\ell_2}-\|e(\mu_t)-e'(\mu_t)\|_{\ell_2}\\
    &\geq (4c+1)\delta_\mu-\delta_\mu\\
    &=4c\delta_\mu
\end{align*}
by the reverse triangle inequality. As we assumed we have that $\|e'(\mu_t)-e(\mu_t)\|_{\ell_2}\leq \delta_\mu$, it follows that $(4c)^{-1}\|e(\lambda)-e(\mu_t)\|_{\ell_2}\geq \|(e'(\mu_t)-e(\lambda))-(e(\mu_t)-e(\lambda))\|_{\ell_2}$. Multiplying the inequality by $\beta$, we see that the conditions of Proposition~\ref{prop:approx_grad} are satisfied for $z(\mu_t):=\beta (e(\lambda)-e'(\mu_t))$. Let us now show the convergence. By our choice of initial point, we have that:
\begin{align*}
    D(\sigma(\lambda)\|\sigma(0))\leq n\log(d)\,.
\end{align*}
Now, suppose that we did not stop before $T$ iterations.
It follows from a telescopic sum argument and Proposition~\ref{prop:approx_grad} that:
\begin{align*}
    D(\sigma(\lambda)\|\sigma(\mu_T))\leq n\log(d)-T\frac{9\beta^{2}(4c+1)^2\delta_\mu^2}{10\,c\,U}\,,
\end{align*}
since $\|e'(\mu_t)-e(\lambda)\|\ge (4c+1)\delta_\mu$ at all iterations because we did not halt. As the relative entropy is positive, it follows that $T=\cO(\beta^{-2}U c^{-1}n\log(d)\delta_\mu^{-2})$ before the stopping criterium is met. The recovery guarantee whenever the stopping criterium is met follows from Proposition~\ref{prop:recoveryguarantee}, the Cauchy-Schwarz inequality and the equivalence of norms $\|\lambda-\mu\|_{\ell_2}\le \sqrt{m}\|\lambda-\mu\|_{\ell_\infty}\le 2 \sqrt{m}$.
\end{proof}
Since we proved in Proposition \ref{upperconstant} that $U=\cO(\beta^2m)$, it follows that the number of iterations is $\cO(nm)$. Thus, we see that having a lower bound on the Hessian is not necessary to ensure convergence, but it can speed it up exponentially: 
\begin{prop}[Exponential convergence of gradient descent with approximate gradients]\label{propLU}
In the same setting as Proposition~\ref{prop:approx_grad} we have:
\begin{align}\label{equ:exponentialconvergence}
    f\lb\mu-\frac{z(\mu)}{cU}\rb-f(\lambda)\leq \lb1-\frac{18L}{10cU}\rb (f(\mu)-f(\lambda)).
\end{align}
In particular, gradient descent with approximate gradient computations starting at $\mu_0=0$ converges after $\cO\lb\frac{U}{L}\log(n\epsilon^{-1})\rb$
iterations to $\mu$ such that $f(\mu)-f(\lambda)\leq \epsilon$.
\end{prop}
\begin{proof}
For any strongly convex function $f$ and points $\mu,\xi\in C$ we have that:
\begin{align*}
f(\xi)\geq f(\mu)+\scalar{\nabla f(\mu)}{(\xi-\mu)}+\frac{L}{2}\|\mu-\xi\|_{\ell_2}^2.
\end{align*}
As explained in~\cite[Chapter 9]{boyd2004convex}, the R.H.S. of the equation above is a convex quadratic function of $\xi$  for $\mu$ fixed.  One can then easily show that its  minimum is achieved at $\tilde{\xi}=\mu-\frac{1}{L}\nabla f(\mu)$. From this we obtain:
\begin{align*}
    f\lb\mu\rb-f(\lambda)\geq-\frac{1}{2L}\|\nabla f(\mu)\|_{\ell_2}^2=-\frac{\beta^2\|e(\mu)-e(\lambda)\|_{\ell_2}^2}{2L}\,,
\end{align*}
where the last identity follows from Eq.~\eqref{firstderive}. By subtracting $f(\lambda)$ from both sides of the inequality~\eqref{equ:convergenceentropy} in Proposition~\ref{prop:approx_grad} and rearranging the terms we have that:
\begin{align*}
    f\lb\mu-\frac{z(\mu)}{cU}\rb-f(\lambda) &\leq f(\mu)-f(\lambda) -\frac{9\beta^2\|e(\mu)-e(\lambda)\|_{\ell_2}^2}{10cU}\\&\leq f\lb\mu\rb-f(\lambda) -\frac{18L}{10cU}(f\lb\mu\rb-f(\lambda))\\
    &=\lb1-\frac{18L}{10cU}\rb(f(\mu)-f(\lambda))\,.
\end{align*}
This yields the claim in Eq.~\eqref{equ:exponentialconvergence}. To obtain the second claim, note that applying Eq.~\eqref{equ:exponentialconvergence} iteratively yields that after $k$ iterations we have that, for $\mu_{k}=\mu_{k-1}-\frac{z(\mu_{k-1})}{cU}$\,,
\begin{align*}
   \lb f(\mu_k)-f(\lambda) \rb\leq \lb1-\frac{18L}{10cU}\rb^k\lb f(0)-f(\lambda)\rb.
\end{align*}
By our choice of initial point and Lemma \ref{propdiffrelent} we have that $f(0)-f(\lambda)=\cO(n)$,
which yields the claim solving for $k$ and noting that $-\log\lb1-\frac{18L}{10cU}\rb= \Omega(\frac{L}{U})$.
\end{proof}

\begin{rem}[Comparison to mirror descent]
It is also worth noting that the convergence guarantees of Proposition~\ref{prop:convergence_grad} and update rules of gradient descent are very similar to the ones of mirror descent with the von Neumann entropy as potential, another algorithm used for learning of quantum states~\cite{Aaronson_2018,Youssry_2019,Brandao2017b,brandao_fast_2020}. In this context, mirror descent would use a similar update rule. However, instead of computing the whole gradient, i.e. all expectation values of the basis, for one iteration, mirror descent just requires us to find one $i$ such that $|e_i(\lambda)-e_i(\mu)|\geq \delta$ and updates the Hamiltonian in the direction $i$. This implies that the algorithm can be run online while we still estimate some other $e_i$, but we will not analyze this variation more in detail here.
\end{rem}

Finally, we assumed so far that we knew the expectation values of the target state, $e(\lambda)$, exactly. However, it follows straightforwardly from Proposition~\ref{prop:convergence_grad} that knowing each expectation value up to an error $\epsilon$ is sufficient to ensure that the additional error due to statistical fluctuations is at most $\epsilon m$. More precisely, if we have that $\|e(\lambda)-e'(\lambda)\|_\infty\leq\epsilon$ for some $\epsilon>0$, then any Gibbs state $\sigma(\mu_*)$ satisfying $\|e(\mu_*)-e'(\lambda)\|_{\ell_2}\leq\delta$ satisfies:
\begin{align*}
    D(\sigma(\lambda)\|\sigma(\mu_*))+D(\sigma(\mu_*)\|\sigma(\lambda))\leq 2\beta \delta\sqrt{m}+2\beta\epsilon m
\end{align*}
by Proposition~\ref{prop:recoveryguarantee} and a Cauchy-Schwarz inequality. With these statements at hand we are finally ready to prove Thm.~\ref{thm:comcomplexitygeneral}.
\begin{proof}[Proof of Thm.~\ref{thm:comcomplexitygeneral}]
We will show in Propositions \ref{prop:convergence_grad}, \ref{propLU} that under the conditions outlined above, the maximum entropy problem will converge to a $\mu_*$ that satisfies:
\begin{align*}
\|e'(\lambda)-e(\mu_*)\|_{\ell_2}\leq (4c+1)\delta_\mu.
\end{align*}
Without making any assumptions on $L$ we can then bound
\begin{align*}
 D(\sigma(\lambda)\|\sigma(\mu_*))+D(\sigma(\mu_*)\|\sigma(\lambda)&= \beta\, |\langle\lambda-\mu_*|e(\lambda)-e(\mu_*)\rangle|\\&\leq \beta\, (|\langle\lambda-\mu_*|e(\lambda)-e(\mu_*)\rangle|+|\langle\lambda-\mu_*|e'(\lambda)-e(\mu_*)\rangle|)\\
 &\leq (4c+1)\delta_{\mu}\sqrt{m}+2\beta\epsilon m
\end{align*}
by H\"older inequality and our assumptions on $e'(\lambda)$. Let us now discuss how strong convexity can improve these estimates. First note that by strong convexity and Cauchy-Schwarz we have:
\begin{align*}
 D(\sigma(\lambda)\|\sigma(\mu_*))+D(\sigma(\mu_*)\|\sigma(\lambda)&= \beta\, |\langle\lambda-\mu_*|e(\lambda)-e(\mu_*)\rangle|\\
 &\leq 
 \beta\|e(\lambda)-e(\mu_*)\|_{\ell_2}\|\lambda-\mu_*\|_{\ell_2}\\
 &\leq L^{-1}\beta^2\|e(\lambda)-e(\mu_*)\|_{\ell_2}^2\\
 &\leq L^{-1}\beta^2\left( \|e'(\lambda)-e(\mu_*)\|_{\ell_2}+\|e(\lambda)-e'(\lambda)\|_{\ell_2}\right)^2
\end{align*} 
which yields the claim.
\end{proof}

In short, we see that we can perform the recovery by simply computing the gradient approximately. In particular, as already hinted at in~\cite{2004.07266}, this implies that recent methods developed to approximately compute the partition function of high-temperature quantum Gibbs states can be used to perform the postprocessing in polynomial time~\cite{kliesch_locality_2014,PhysRevLett.124.220601,harrow2020classical,kuwahara2020gaussian}. %
This and other methods to compute the gradient are discussed in more detail in Sec.~\ref{sec:regimesefficient}.
Furthermore, it should be noted that usually $L=\Omega(\beta^{-2})$ in the high temperature regime, making the bound independent of $\beta$ for such states. We refer to Sec.~\ref{sec:strong_convexity} for a summary of the cases for which bounds on $L$ are known.

\section{Lipschitz constants and transportation cost inequalities}\label{Transportcost}
In this section, we identify conditions under which it is possible to estimate all expectation values of $k$-local observables up to an error $\epsilon$ by measuring $\cO(\textrm{poly}(k,\log(n),\epsilon^{-1}))$ copies of it, where $n$ is the system size, which constitutes an exponential improvement in some regimes.
To obtain this result, we combine the maximum entropy method introduced in Section \ref{sec:max-entropyprin} with techniques from quantum optimal transport. 
In order to formalize and prove the result claimed, we resort to transportation cost inequalities and the notion of Lipschitz constants of observables, which we now introduce.

\subsection{Lipschitz constants and Wasserstein metrics}\label{sec:lipsch}
Transportation cost inequalities, introduced by Talagrand in the seminal paper~\cite{Talagrand_1996},  constitute one of the strongest tools available to show concentration of measure inequalities. In the quantum setting, their study was initiated in~\cite{carlen2014analog,carlen2017gradient,Rouz2019,palma_optimal}. Here we are going to show how they can also be used in the context of quantum tomography. On a high level, a transportation cost inequality for a state $\sigma$ quantifies by how much the relative entropy with respect to another state $\rho$ is a good proxy to estimate to what extent the expectation values of sufficiently regular observables differ on the states. As maximum entropy methods allow for a straightforward control of the convergence of the learning in relative entropy (cf. Section \ref{sec:max-entropyprin}), they can be combined to derive strong recovery guarantees. But first we need to define what we mean by a regular observable.

We start by a short discussion of Lipschitz constants and the Wasserstein-1 distance. To obtain an intuitive grasp of these concepts, one way is to first recall the variational formulation of the trace distance of two quantum states $\sigma,\rho$:
\begin{align*}
\|\rho-\sigma\|_{tr}=\sup\limits_{P=P^\dagger, \|P\|_\infty\leq1}\tr{P(\rho-\sigma)}.
\end{align*}
Seeing probability distributions as diagonal quantum states, we recover the variational formulation of the total variation distance by noting that we may restrict to diagonal operators $P$. Thus, the total variation distance quantifies by how much the expectation values of arbitrary bounded functions can differ under the two distributions. However, in many situations we are not interested in expectation values of arbitrary bounded observables, but rather observables that are sufficiently regular. E.g., most observables of physical interest are (quasi)-local. Thus, it is natural to look for distance measures between quantum states that capture the notion that two states do not differ by much when restricting to expectation values of sufficiently regular observables. These concerns are particularly relevant in the context of tomography protocols, as they should be designed to efficiently obtain a state that reflects the expectation values of extensive observables of the system. As we will see, one of the ways of ensuring that the sample complexity of the tomography algorithm reflects the regularity of the observables we wish to recover is through demanding a good recovery in the Wasserstein distance of order $1$~\cite{Rouz2019,palma_optimal}.

In the classical setting~\cite[Chapter 3]{raginsky_concentration_2014}, one way to define a Wasserstein-1 distance between two probability distributions is by replacing the optimization over all bounded diagonal observables by that over those that are sufficiently regular: given a metric $d$ on a sample space $\Omega$, we define the Lipschitz constant of a function $f:\Omega\to\R$ to be:
\begin{align*}
\|f\|_{\operatorname{Lip}}:=\sup\limits_{x,y\in\Omega}\frac{|f(x)-f(y)|}{d(x,y)}\,.
\end{align*}
Denoting the Wasserstein-1 distance by  $W_1$, it is given for two probability measures $p,q$ on $\Omega$ by
\begin{align}\label{equ:variationalw1}
W_1(p,q):=\sup_{f:\|f\|_{\operatorname{Lip}}\leq1}|\mathbb{E}_p(f)-\mathbb{E}_q(f)|\,.
\end{align}
That is, this metric quantifies by how much the expectation values of sufficiently regular functions can vary under $p$ and $q$, in clear analogy to the variational formulation of the trace distance. We refer to~\cite{raginsky_concentration_2014,villani_optimal_2009} for other interpretations and formulations of this metric.

\subsubsection{Quantum Hamming Wasserstein distance}~\label{sec:hammingwass}

It is not immediately clear how to generalize these concepts to noncommutative spaces. There are by now several definitions of transport metrics for quantum states~\cite{Rouz2019,palma_optimal,kiani2021quantum,giacomo_cambyse}. As already noted in the main text, de Palma et al. defined the Lipschitz constant of an observable $O\in\M_{d^n}$ as \cite{palma_optimal}:
\begin{align}\label{equ:definitiongiacomo2}
    \|O\|_{\operatorname{Lip},\square}=\sqrt{n}\max\limits_{1\leq i\leq n}\underset{\substack{\rho,\sigma\in\D_{d^n}\\\operatorname{tr}_i[\rho]=\operatorname{tr}_i[\sigma]}}{\max}\,\tr{O(\rho-\sigma)}\,.
\end{align}
That is, the Lipschitz constant quantifies the amount by which the expectation value of an observable can change when evaluated on two states that only differ on one site. This is in analogy with the Lipschitz constants induced by the Hamming distance on the hypercube, so we denote it with $\square$. Note that in our definition we added the $\sqrt{n}$  factor, which will turn out to be convenient later. Armed with this definition, we can immediately obtain an analogous definition of the Wasserstein distance in Eq.~\eqref{equ:variationalw1} for two states $\rho,\sigma$:
\begin{align}\label{equ:variationalw1quantumgiacomo}
{W}_{1,\square}(\rho,,\sigma):=\sup_{O=O^\dagger,\|O\|_{\operatorname{Lip},\square}\leq1}\tr{O(\rho-\sigma)}\,.
\end{align}
The authors of \cite{palma_optimal} also put forth the following equivalent expression for the norm: 
\begin{align}\label{wassersteinnormDP}
{W}_{1,\square}(\rho,\sigma)=\frac{\min \left\{\sum_{i=1}^{n} \|X^{(i)}\|_1:\rho-\sigma=\sum_{i=1}^{n} X^{(i)}, X^{(i)} \in \mathcal{M}_{d^n}^{\operatorname{sa}}, \operatorname{tr}_{i} [X^{(i)}]=0\right\}}{2\sqrt{n}}\,.
\end{align}
It follows from an application of H\"older's inequality combined with the variational formulation in Eq.~\eqref{equ:definitiongiacomo2} that $\|O\|_{\operatorname{Lip},\square}\leq 2\sqrt{n}\|O\|_\infty$. However, it can be the case that $\|O\|_{\operatorname{Lip}}\ll \sqrt{n}\,\|O\|_\infty$, which are exactly those observables that should be thought of as regular. This is because this signals that changing the state locally leads to significantly smaller changes to the expectation value of the observable than global ones. Two examples of this behavior are given by the observables:
\begin{align*}
O_1=\sum\limits_{i=1}^n Z_{i^c},\quad ~~~~~\text{ and }~~~~~ \quad O_2=\sum\limits_{i=1}^n Z_i\,,
\end{align*}
where $Z_{i^c}$ acts as identity on $i$ and $Z_i$ else, i.e. $Z_{1^c}=I_1\otimes Z_2\otimes Z_3\otimes \cdots\otimes Z_n$.
Clearly, $\|O_1\|_\infty=\|O_2\|_\infty=n$.
On the other hand, by considering the states $\rho=\ketbra{0}^{\otimes n}$ and $\sigma=\ketbra{1}\otimes\ketbra{0}^{\otimes n-1}$, we see that $\|O_1\|_{\operatorname{Lip},\square}\geq\sqrt{n}\,( 2n-2)$ whereas  $\|O_2\|_{\operatorname{Lip},\square}=2\sqrt{n}$. More generally, it is not difficult to see that if $O=\sum\limits_{i=1}^nO_i$ with $\|O_i\|_\infty\le 1$ and we denote by $\operatorname{supp}(O_i)$ the qudits on which $O_i$ acts nontrivially, then:
\begin{align*}
\|O\|_{\operatorname{Lip},\square}\leq 2\sqrt{n}\max_{1\leq j\leq n}\sum_i|\operatorname{supp}(O_i)\cap\{j\}|\,.
\end{align*}
That is, the maximal number of intersections of the support of $O_i$ on one qubit. From these examples we see that for local observables, the ratio of the operator norm and Lipschitz constant reflects the locality of the observable.  

\subsubsection{Quantum differential Wasserstein distance}\label{sec:diffwass}

The Wasserstein distance $W_{1,\square}$ generalizes the classical Orstein distance, that is the Wasserstein distance corresponding to the Hamming distance on bit strings. Another definition of a Lipschitz constant and its attached Wasserstein distance was put forth in~\cite{Rouz2019},
where the construction is based on a differential structure that bears more resemblance to that of the Lipschitz constant of a differentiable function on a continuous sample space, e.g. a smooth manifold~\cite{raginsky_concentration_2014}. Let us now define the notion of a noncommutative differential structure (see \cite{Carlen_2019}):

\begin{defi}[Differential structure]\label{thm:normalformCM}
A set of operators $L_k \in\M_{d^n}$ and constants $\omega_k\in\R$ defines a differential structure $\{L_k,\omega_k\}_{k\in\mathcal{K}}$ for a full rank state $\sigma\in\D_{d^n}$ if
					\begin{itemize}
						\item[1] $\{L_k\}_{k\in\mathcal{K}}=\{L_k^{\dagger}\}_{k\in\mathcal{K}}$;
						\item[2] $\{L_k\}_{k\in\mathcal{K}}$ consists of eigenvectors of the modular operator $\Delta_\sigma(X):=\sigma X\sigma^{-1}$ with
						\begin{align}\label{eigenD}
							\Delta_\sigma(L_k)=e^{-\omega_k}L_k\,.
						\end{align}
						\item [3] $\|L_k\|_{\infty}\leq 1$.
						\end{itemize}

\end{defi}
Such a differential structure can be used to provide the set of matrices with a Lipschitz constant that is tailored to $\sigma$, see e.g.~\cite{Rouz2019,Carlen_2019} for more on this. In order to distinguish that constant from the one defined in~\eqref{equ:definitiongiacomo2}, we will refer to it as the differential Lipschitz constant and denote it by $\|X\|_{\operatorname{Lip},\nabla}$.
It is given by:
\begin{align}\label{equ:lipnorm}
\|X\|_{\operatorname{Lip},\nabla}:= \left(  \sum_{k\in\mathcal{K}}  (e^{-\omega_k/2}+e^{\omega_k/2})\|[L_k,X]\|_{\infty}^2\right)^{1/2}\,.
\end{align}
The quantity $[L_k,X]$ should be interpreted as a partial derivative and is sometimes denoted by $\partial _kX$ for that reason. Then, the gradient of a matrix $A$, denoted by $\nabla A$ with a slight abuse of notations, refers to the vector of operator-valued coordinates $(\nabla A)_i=\partial_i A$. For ease of notation, we will denote the differential structure by the couple $(\nabla,\sigma)$. The notion of a differential structure is also intimately connected to that of the generator of a quantum dynamical semigroup converging to $\sigma$~\cite{Carlen_2019}, and properties of that semigroup immediately translate to properties of the metric. This is because the differential structure can be used to define an operator that behaves in an analogous way to the Laplacian on a smooth manifold, which in turn induces the heat semigroup. We refer to~\cite{Carlen_2019,Rouz2019} for more details.

To illustrate the differential structure version of the Lipschitz constant, it is instructive to think of the maximally mixed state. In this case, one possible choice would consist of picking the $L_k$ to be all $1-$local Pauli strings and $\omega_j=0$. 
Then the Lipschitz constant turns out to be given by:
\begin{align}\label{equ:lipnormdepo}
\|X\|_{\operatorname{Lip},\nabla}=\left(  \sum_{k\in\mathcal{K}} \|P_kXP_k-X\|_{\infty}^2\right)^{1/2},
\end{align}
where $P_k$ are all $1-$local Pauli matrices. Thus, we see that this measures how much the operator changes if we act locally with a Pauli unitary on it. If we think of an operator as a function and conjugating with a Pauli as moving in a direction, the formula above indeed looks like a derivative. In fact, it is possible to make this connection rigorous, see~\cite{Carlen_2019}.

As before, the definition in Eq.~\eqref{equ:lipnorm} yields a metric on states by duality:
\begin{align*}
    W_{1,\nabla}(\rho,\sigma):=\sup\limits_{X=X^\dagger,\, \|X\|_{\operatorname{Lip},\nabla}\leq1}\left|\operatorname{Tr}\left(X(\rho-\sigma)\right)\right|.
\end{align*}
It immediately follows from the definitions that for any observable $X$:
\begin{align}\label{equ:Lipschitz}
    \left|\tr{X(\rho-\sigma)}\right|\leq \|X\|_{\operatorname{Lip},\nabla} \, W_{1,\nabla}(\rho,\sigma) \, .
\end{align}

Although this geometric interpretation opens up other interesting mathematical connections for this definition, the differential Wasserstein distance has the downside of being state dependent. It however induces a stronger topology than the quantum Hamming Wasserstein distance in some situations (see \cite[Proposition 5]{giacomo_cambyse})). In particular, the results of \cite[Proposition 5]{giacomo_cambyse}) imply that for commutative Gibbs states a TC inequality for $ W_{1,\nabla}$ implies the corresponding statement for $ W_{1,\square}$.

\subsection{Local evolution of Lipschitz observables}\label{sec:examplesLipsch}
As already discussed in Subsections~\ref{sec:hammingwass} and~\ref{sec:diffwass} when we defined $\|.\|_{\operatorname{Lip},\nabla}$ and $\|.\|_{\operatorname{Lip},\square}$, Lipschitz constants can be easily controlled by making assumptions on the locality of the operators. Indeed, if we apply a local circuit to a local observable, it is straightforward to control the growth of the Lipschitz constant in terms of the growth of the support of the observable under the circuit. More precisely, in~\cite[Proposition 13]{palma_optimal} the authors show such a statement for discrete time evolutions with exact lightcones: if we denote by $|L|$ the size of the largest lightcone of one qubit under a channel $\Phi$, then for any observable $O\in\D_{d^n}$, $\|\Phi^*(O)\|_{\operatorname{Lip},\square}\leq 2|L|\|O\|_{\operatorname{Lip},\square}$. Here we will extend such arguments to the evolution under local Hamiltonians or Lindbladians. By resorting to Lieb-Robinson bounds, we show that the Lipschitz constants $\|.\|_{\operatorname{Lip},\nabla}$ and $\|.\|_{\operatorname{Lip},\square}$ of initially local observables evolving according to a quasi-local dynamics increase at most with the effective lightcone of the evolution. Thus, short-time dynamics and shallow-depth quantum channels can only mildly increase the Lipschitz constant. This further justifies the intuition that observables with small Lipschitz constant reflect physical observables.

Lieb-Robinson (LR) bounds in essence assert that the time evolution of local observables under (quasi)-local short-time dynamics have an effective lightcone. 
There are various formulations of Lieb-Robinson bounds. Reviewing those in detail is beyond the scope of this work and we refer to~\cite{hastings_locality_2010,poulin_lieb-robinson_2010,barthel2012quasilocality,bach_lieb-robinson_2014} and references therein for more details. For studying the behaviour of $\|.\|_{\operatorname{Lip},\nabla}$ under local evolutions, the most natural formulation is the \emph{commutator version}: the generator $\cL$ of a quasi-local dynamics $t\mapsto \Phi_t=e^{t\cL}$ on  $n$ qudits arranged on a graph $G=(V,E)$, with graph distance $\operatorname{dist}$, is said to satisfy a LR bound with LR velocity $v$ if for any observable $O_A$ supported on a region $A$ and any other observable $B$ supported on a region $B$, we have:
\begin{align}\label{equ:lieb-robinson}\tag{LR1}
\|[\Phi_t^*(O_A),O_B]\|_\infty\leq c\,(e^{ vt}-1)\,g(\operatorname{dist}(A,B))\,\|O_A\|_\infty\|O_B\|_\infty\,,
\end{align}
for $g:\mathbb{N}\to \mathbb{R}_+$ a function such that $\lim_{x\to\infty}g(x)=0$. We then have:
\begin{prop}[Growth of differential Lipschitz constant for local evolutions]\label{prop:lipdiff}
Let $(\nabla,\sigma)$ be a differential structure on $\M_{d^n}$ and let $O=\sum\limits_{i}O_i$ be an observable with $\|O_i\|_{\infty}\leq 1$. Let $A_i$ denote the support of each $O_i$ and $B_j$ that of each $L_j$. Moreover, let $t\mapsto \Phi_t$ be an evolution satisfying Eq.~\eqref{equ:lieb-robinson} and set $o(i,j)(t)=2$ if $A_i\cap B_j\not=\varnothing$ and $c\,(e^{ vt}-1)\,g(\operatorname{dist}(A_i,B_j))$ else. Then:
\begin{align*}
\|\Phi_t^*(O)\|_{\operatorname{Lip},\nabla}^2\leq \sum_{k\in\mathcal{K}}(e^{\omega_j}+e^{-\omega_j})\lb \sum_io_{i,j}(t)\rb^2.
\end{align*}
\end{prop}
\begin{proof}
The proof follows almost by definition. We have:
\begin{align*}
\|\Phi_t^*(O)\|_{\operatorname{Lip},\nabla}^2=\sum_{j\in\mathcal{J}}(e^{\omega_j}+e^{-\omega_j})\|[\Phi_t(O),L_j]\|_{\infty}^2.
\end{align*}
By a triangle inequality we have that: 
\begin{align*}
\|[\Phi_t^*(O),L_i]\|_{\infty}^2\leq\lb\sum_i\|[\Phi_t^*(O_i),L_j]\|_{\infty}\rb^2
\end{align*}
For any term in the sum we have $\|[\Phi_t^*(O_i),L_j]\|_\infty\leq2$ by the submultiplicativity of the operator norm, a triangle inequality and $\|L_j\|_\infty\le 1$.  In case $O_i$ and $L_j$  do not overlap, the stronger bound in Eq.~\eqref{equ:lieb-robinson} holds.
\end{proof}
To illustrate this bound more concretely, let us take $O=\sum_{i=1}^nZ_i$, $L_j$ acting on $[j,j+k]$ for $j=1,\ldots, n-k$, and $g(\operatorname{dist}(i,j))=e^{-\mu|i-j|}$, for some constant $\mu$. I.e. we have a $1$D differential structure and a local time evolution on a $1$D lattice. Then for any $j$:
\begin{align*}
\sum_io_{i,j}(t)=k+(e^{vt}-1)\lb\sum_{i=1}^{j-1}e^{-\mu|i-j|}+\sum_{i=j+k+1}^{n}e^{-\mu|i-j|}\rb\leq k+\frac{(e^{vt}-1)e^
{-\mu}}{1-e^{-\mu}}.
\end{align*}
Thus,
\begin{align*}
\|\Phi_t^*(O)\|_{\operatorname{Lip},\nabla}\leq \sqrt{n-k}\lb k+\frac{(e^{vt}-1)e^
{-\mu}}{1-e^{-\mu}}\rb.
\end{align*}
We see that for constant times the Lipschitz constant is still of order $\sqrt{n}k$.

Let us now derive a similar, yet somehow stronger, version of Prop.~\ref{prop:lipdiff} for $\|.\|_{\operatorname{Lip},\square}$. In some situations, bounds like \eqref{equ:lieb-robinson} can be further exploited in order to prove the quasi-locality of Markovian quantum dynamics \cite{barthel2012quasilocality}: for any region $C\subset D\subset V$, there exists a local Markovian evolution $t\mapsto \Phi_t^{(D)}$ that acts non-trivially only on region $D$, and such that for any observable $O_C$ supported on region $C$, 
\begin{align}\tag{LR2}\label{LR2}
    \|\big(\Phi^*_t-(\Phi_t^{(D)})^*\big)(O_C)\|_\infty\le c'\,(e^{ vt}-1)\,h(d(C,V\backslash D))\,\|O_C\|_\infty\,,
\end{align}
for some other function $h:\mathbb{N}\to \mathbb{R}_+$ such that $\lim_{x\to\infty}h(x)=0$ and constant $c'>0$. In other words, at small times, the channels $\Phi_t^*$ can be well-approximated by local Markovian dynamics when acting on local observables. Let us now estimate the growth of Lipschitz constants for the definition of~\cite{palma_optimal} given a Lieb-Robinson bound:
\begin{prop}
Assume that $\Phi_t$ satisfies the bound \eqref{LR2}. Then, for any two quantum states $\rho,\sigma\in \D_{d^n}$ and any ordering $\{1,\dots,n\}$ of the graph:
\begin{align}\label{LRboundsW1}
    W_{1,\square}(\Phi_t(\rho),\Phi_t(\sigma))\le \Big(8+2\,c'\,(e^{ vt}-1)\,\sum_{i=3}^nh(d(\{i\cdots n\},\{1\}))\Big)\,W_{1,\square}(\rho,\sigma)\,.
\end{align}
Moreover, for any observable $H\in\M_{d^n}$,
\begin{align}\label{LRboundLip}
  \|\Phi_t^*(H)\|_{\operatorname{Lip},\square}\le \Big(8+2\,c'\,(e^{ vt}-1)\,\sum_{i=3}^nh(d(\{i\cdots n\},\{1\}))\,\Big)\,\|H\|_{\operatorname{Lip},\square}\,.
\end{align}
\end{prop}

\begin{proof}
From \cite{palma_optimal}, the Wasserstein distance $W_{1,\square}$ arises from a norm $\|.\|_{W_1}$, i.e. $W_{1,\square}(\rho,\sigma)=\|\rho-\sigma\|_{W_1}$. Moreover, the norm $\|.\|_{W_1}$ is uniquely determined by its unit ball $\mathcal{B}_n$, which in turn is the convex hull of the set of the differences between couples of neighboring quantum states: 
\begin{align*}
    \mathcal{N}_n=\bigcup_{i\in V}\,\mathcal{N}_n^{(i)}\,,~~~~~\mathcal{N}_n^{(i)}=\{\rho-\sigma:\,\rho,\sigma\in\D_{d^n},\,\operatorname{tr}_i(\rho)=\operatorname{tr}_i(\sigma)\}\,.
    \end{align*}
    Now by convexity, the contraction coefficient for this norm is equal to
    \begin{align*}
        \|\Phi_t\|_{W_1\to W_1}=\max\big(\|\Phi_t(X)\|_{W_1}:\,X\in \M_{d^n}^{\operatorname{sa},0},\,\|X\|_{W_1}\le 1  \big)=\max_{X\in\mathcal{N}_n}\|\Phi_t(X)\|_{W_1}\,,
    \end{align*}
    where $\M_{d^n}^{\operatorname{sa},0}$ denotes the set of self-adjoint, traceless observables. Let then $X\in\mathcal{N}_n$. By the expression \eqref{wassersteinnormDP}, and choosing without loss of generality an ordering of the vertices such that $\operatorname{tr}_1(X)=0$, we have 
    \begin{align}\label{equ:LR_estimates}
        \|\Phi_t(X)\|_{W_1}&\le \frac{1}{2\sqrt{n}}\sum_{i=1 }^{n}\Big\|\frac{I}{d^{i-1}}\otimes \operatorname{tr}_{1\cdots i-1}\circ \Phi_t(X)-\frac{I}{d^i}\otimes \operatorname{tr}_{1\cdots i}\circ \Phi_t(X)\Big\|_1\nonumber\\
        &=\frac{1}{2\sqrt{n}}\,\sum_{i=1}^n\,\Big\|\,\int\,d\mu(U_i)\,\operatorname{tr}_{1\cdots i-1}\circ (\Phi_t(X)-U_i\Phi_t(X)U_i^\dagger)\Big\|_1\nonumber\\
        &\le \frac{1}{2\sqrt{n}}\sum_{i=1}^n\int\,d\mu(U_i)\,\|[U_i,\operatorname{tr}_{1\cdots i-1}\circ\Phi_t(X)]\|_1\nonumber\\
        &\le \frac{1}{\sqrt{n}}\,\sum_{i=1}^n\|\operatorname{tr}_{1\cdots i-1}\circ\Phi_t(X)\|_1\nonumber\\
        &\overset{(1)}{=} \frac{1}{\sqrt{n}}\,\sum_{i=1}^n\|\operatorname{tr}_{1\cdots i-1}\circ(\Phi_t-\Phi^{(i-k\cdots n)}_t)(X)\|_1
        \end{align}
where $\mu$ denotes the Haar measure on one qudit, and where $(1)$ follows from the fact that $\operatorname{tr}_1(X)=0$, with $\Phi_t^{(i-k\cdots n)}\equiv \Phi_t^{(\{i-k,\cdots, n\})}$ defined as in Eq. \eqref{LR2} with $k<i-1$. Next, by the variational formulation of the trace distance and Eq. \eqref{LR2}, we have for $i\geq3$ that
\begin{align*}
    \|\operatorname{tr}_{1\cdots i-1}\circ(\Phi_t-\Phi^{(i-k\cdots n)}_t)(X)\|_1&=\max_{\|O_{i\cdots n}\|_\infty\le 1}\,\big|\operatorname{tr}\big[X(\Phi_t^*-\Phi^{(i-k\cdots n)*}_t)(O_{i\cdots n}) \big]\big|\\
    &\le \,\max_{\|O_{i\cdots n}\|_\infty\le 1}\,\|(\Phi_t^*-\Phi^{(i-k\cdots n)*}_t)(O_{i\cdots n}) \|_\infty\|X\|_1\\
    &\le \, c'\,(e^{ vt}-1)\,h(\operatorname{dist}(\{i\cdots n\},\{1\cdots i-k-1\}))\,\|X\|_1\\
    &\overset{(2)}{\le} 2\, c'\,(e^{ vt}-1)\,h(\operatorname{dist}(\{i\cdots n\},\{1\cdots i-k-1\}))\,\sqrt{n}\,\|X\|_{W_1}\,,
    \end{align*}
    where $(2)$ follows from \cite[Proposition 6]{palma_optimal}. By picking $k=i-2$ and inserting this estimate into Eq.~\eqref{equ:LR_estimates} for $i\geq3$ and the trivial  estimate $\|\operatorname{tr}_{1\cdots i-1}\circ(\Phi_t-\Phi^{(i-k\cdots n)}_t)(X)\|_1\leq2\|X\|_1$ for $i=1,2$ we obtain Eq. \eqref{LRboundsW1}. Eq. \eqref{LRboundLip} follows by duality.
\end{proof}

\subsection{Transportation cost inequalities}\label{sec:TCIs}
Although interesting on their own, the relevance of the Lipschitz constants introduced above becomes clearer in our context when we also have a \emph{transportation cost} inequality~\cite{gozlan2010transport,raginsky_concentration_2014}. A quantum state $\sigma$ satisfies a transportation cost inequality with constant $\alpha>0$ if for any other state $\rho$ it holds that:
\begin{align}\label{equ:transportation_entropy_inequality}
   W_{1}(\rho,\sigma)\leq \sqrt{\frac{1}{2\alpha}\,D(\rho\|\sigma)}\,,
\end{align}
where $W_1\in\{{W}_{1,\square},W_{1,\nabla},W_{1,\operatorname{loc}}\}\,.$ In what follows, we simply write $\|.\|_{\operatorname{Lip}}$ and $W_1$ to denote either of the Lipschitz constants, and their corresponding Wasserstein metrics, defined above. This inequality should be thought of as a stronger version of Pinsker's inequality that is tailored to a state $\sigma$ and the underlying Wasserstein distance.

One of the well-established techniques to establish a transportation cost inequality for $W_{1,\nabla}$ is by exploiting the fact that it is implied by a functional inequality called the modified logarithmic Sobolev inequality. It is beyond the scope of this paper to explain this connection and we refer to e.g.~\cite{Rouz2019} and references therein for a discussion on these topics. But for our purposes it is important to note that in~\cite{capel2020modified} the authors and Capel show modified logarithmic Sobolev inequalities for several classes of Gibbs states. More recently, one of the authors and De Palma derived such transportation cost inequalities for $W_{1,\square}$ in \cite{giacomo_cambyse}. In Theorem \ref{thm:modifiedmlsi} below we summarize the regimes for which transportation cost inequalities are known to hold:

\begin{theorem}[transportation cost for commuting Hamiltonians~\cite{capel2020modified,giacomo_cambyse,MLSI1D}]\label{thm:modifiedmlsi} 
Let $E_1,\ldots, E_{m}\subset \M_{d^n}$ be a set of $k$-local linearly independent {commuting} operators with $\|E_i\|_{\infty}\leq 1$. Then $\sigma(\lambda)$ satisfies a transportation cost inequality with constant $\alpha>0$ for all $\lambda\in B_{\ell_{\infty}}(0,1)$ in the following cases:
\begin{itemize}
\item[$\operatorname{(i)}$] The $E_i$ are classical or nearest neighbour (i.e.~$k=2$) on a regular lattice and the inverse temperature $\beta<\beta$ where $\beta_c$ only depends on $k$ and the dimension of the lattice, for both $W_{1,\nabla},W_{1,\square}$ and $\alpha=\Omega(1)$~\cite{capel2020modified}.
\item[$\operatorname{(ii)}$] The operators $E_i$ are local with respect to a hypergraph $G=(V,E)$ and the inverse temperature satisfies $\beta<\beta_c$, where $\beta_c$ only depends on $k$ and properties of the interaction hypergraph for $W_{1,\square}$ and $\alpha=\Omega(1)$~\cite[Theorem 3, Proposition 4]{giacomo_cambyse}.
\item[$\operatorname{(iii)}$] The $E_i$ are one-dimensional and $\beta>0$, for both $W_{1,\nabla}$ and $W_{1,\square}$ and $\alpha=\Omega(\log(n)^{-1})$ \cite{MLSI1D}.%
\end{itemize}
Moreover, the underlying differential structure $(\nabla,\sigma)$ consists of $L_k$ acting on at most $\cO(k)$ qudits.
\end{theorem}
Theorem~\ref{thm:modifiedmlsi} establishes that transportation cost inequalities are satisfied for most classes of commuting Hamiltonians known to have exponential decay of correlations.

\begin{rem}
In~\cite[Proposition 5]{giacomo_cambyse}, the authors show that $W_{1,\square}\leq c(k)\,W_{1,\nabla}$ holds up to some constant $c(k)$ depending on the locality of the differential structure. This implies that a transportation cost inequality for $W_{1,\nabla}$ implies one for $W_{1,\square}$ up to $c(k)$. Thus, although the authors of~\cite{capel2020modified,MLSI1D} only obtain the result for $W_{1,\nabla}$, we can use it to translate it to $W_{1,\square}$. We conclude  that for commuting Hamiltonians TC are available for essentially all classes of local Hamiltonians for which they are expected to hold.
\end{rem}

\section{Combining the maximum entropy method with transportation cost inequalities}\label{sec:Wasslearning}

With these tools at hand, we are now ready to show that by resorting to transportation cost inequalities it is possible to obtain exponential improvements in the sample complexity required to learn a Gibbs state. First, let us briefly review shadow tomography or Pauli regrouping techniques~\cite{Huang2020,cotler2020quantum,jena_pauli_2019,crawford_efficient_2020}. Although these methods all work under slightly different assumptions and performance guarantees, they have in common that they allow us to learn the expectation value of $M$ $k$-local observables $O_1,\ldots,O_M\in\M_{2^n}$  such that $\|O_i\|_{\infty}\leq1$ up to an error $\epsilon$  and failure probability at most $1-\delta$  by measuring $\cO(e^{\cO(k)}\log(M\delta^{-1})\epsilon^{-2})$ copies of the state. 

For instance, for the shadow methods of~\cite{Huang2020}, we obtain a $\cO(4^k\log(M\delta^{-1})\epsilon^{-2})$ scaling by measuring in $1$-qubit random bases. 
The estimate is then obtained by an appropriate median-of-means procedure for the expectation value of each output  string.
The computation for obtaining the expectation value of $E_i$ through this method then entails evaluating the expectation value of the observables on $\cO(4^k\log(M\delta^{-1})\epsilon^{-2})$ product states.  For $k$-local observables $E_i$, evaluating the expectation value of $E_i$ on a product state takes time $\cO(e^{c k}\log(M\delta^{-1})\epsilon^{-2})$ for some $c>0$. Thus, we see that for $k=\cO(\log(n))$ also the postprocessing can be done efficiently.

The application of such results to maximum entropy methods is then clear: given $E_1,\ldots,E_m$ assumed to be at most $k$-local, with probability at least $1-\delta$ we can obtain an estimate $e'(\lambda)$ of $e(\lambda)$ satisfying:
\begin{align}\label{equ:guarantee_shadows}
\|e'(\lambda)-e(\lambda)\|_{\ell_p}\leq m^{\frac{1}{p}}\epsilon
\end{align}
using $\cO(4^k\log(m\delta^{-1})\epsilon^{-2})$ copies of $\sigma(\lambda)$. We then finally obtain our main theorem, Theorem \ref{thn:learning}, restated here for the sake of clarity:

\begin{thm}[Fast learning of Gibbs states]\label{thm:quick_dirty}
Let $\sigma(\lambda)\in\D_{2^n}$ be an $n$-qubit Gibbs state at inverse temperature $\beta$ with respect to a set of $k$-local operators $\mathcal{E}=\{E_i\}_{i=1}^m$ that satisfies a transportation cost inequality with constant $\alpha$. Moreover,  assume that $\mu\mapsto \log(\mathcal{Z}(\mu))$ is $L,U$ strongly  convex in $ B_{\ell_\infty}(0,1)$.
Then 
\begin{align*}
\cO\lb 4^k\alpha^{-1}\epsilon^{-2}\beta\log(m\delta^{-1})\min\left\{\frac{m\beta}{L n},\frac{m^2}{\alpha\,n^2\eps^2}\right\}\rb
\end{align*}
samples of $\sigma(\lambda)$ are sufficient to obtain a state $\sigma(\mu)$ that satisfies:
\begin{align}\label{equ:operatordensity_recovery}
|\tr{O(\sigma(\lambda)-\sigma(\mu)}|\leq \epsilon\, n^{\frac{1}{2}}\|O\|_{\operatorname{Lip}}
\end{align}
for all Lipschitz observables $O$ with probability at least $1-\delta$. 
\end{thm}
\begin{proof}
Using the aforementioned methods of~\cite{Huang2020} we can obtain an estimate $e'(\lambda)$ satisfying the guarantee of Eq.~\eqref{equ:good_estimate} with probability at least $1-\delta$. We will now resort to the results of Thm.~\ref{thm:comcomplexitygeneral} to obtain guarantees on the output of the maximum entropy algorithm.
Now we solve the maximum entropy problem with $e'(\lambda)$ and set the stopping criterion for the output $\mu_*$ as
\begin{align*}
\|e'(\lambda)-e(\mu_*)\|_{\ell_2}< (4c+1)\,\epsilon\sqrt{m}.
\end{align*}
for $c>10$.
Then it follows from Thm.~\ref{thm:comcomplexitygeneral} that:
\begin{align}\label{equ:relative_entropy_final}
D(\sigma(\mu_*)\|\sigma(\lambda))\leq D(\sigma(\mu_*)\|\sigma(\lambda))+D(\sigma(\lambda)\|\sigma(\mu_*))=\cO\left(\beta\min\{{\beta} L^{-1}\epsilon^2m,\epsilon m\}\right).
\end{align}
This can then be combined with transportation cost inequalities. Indeed, we have:
\begin{align*}
|\tr{O(\sigma(\lambda)-\sigma(\mu_*)}|\leq \|O\|_{\operatorname{Lip}}W_1(\sigma(\mu_*),\sigma(\lambda))\leq \|O\|_{\operatorname{Lip}}\sqrt{\frac{D(\sigma(\mu_*),\sigma(\lambda))}{2\alpha}}.
\end{align*}
Inserting our bound on the relative entropy in Eq.~\eqref{equ:relative_entropy_final} we obtain:
\begin{align*}
|\tr{O(\sigma(\lambda)-\sigma(\mu_*)}|= \cO\left(\|O\|_{\operatorname{Lip}}\sqrt{\beta}\min\{ (\alpha L)^{-\frac{1}{2}}\epsilon \sqrt{\beta\,m},\sqrt{\epsilon m\alpha^{-1}}\}\right).
\end{align*}
We conclude by suitably rescaling $\epsilon$.
\end{proof}
The theorem above yields exponential improvements in the sample complexity to learn the expectation value of certain observables for classes of states that satisfy a transportation cost inequality with $\alpha=\Omega(\log(n)^{-1})$. As discussed in Sec.~\ref{sec:examplesLipsch}, extensive observables that can be written as a sum of $l$-local observables have a Lipschitz constant that satisfies $\|O\|_{\operatorname{Lip}}=\cO(l\sqrt{n})$. Shadow-like methods would require 
$\cO(e^{\cO(l)}\log(m\delta^{-1})\,\epsilon^{-2})$ samples to learn such observables up to a relative error of $n\epsilon$. Our methods, however, require 
 $\cO(e^{\cO(k)}\operatorname{poly}(l,\epsilon^{-1})\log(m\delta^{-1}))$, which yields exponential speedups in the regime 
 $l=\operatorname{poly}(\log(n))$. Of course it should also be mentioned that classical shadows do not require any assumptions on the underlying state.

Furthermore, considering the exponential dependency of the sample complexity in the locality for shadow-like methods, we believe that our methods yield practically significant savings already in the regime in which we wish to obtain expectation values of observables with relatively small support. For instance, for high-temperature Gibbs state of nearest neighbor Hamiltonians and observables supported on $15$ qubits, shadows require a factor of $\sim 10^7$ more samples than solving the maximum entropy problem for obtaining the same precision.

On the other hand, previous guarantees on learning quantum many-body states~\cite{Cramer2010,Baumgratz_2013,Lanyon2017,PhysRevA.101.032321} required a polynomial in system size precision to obtain a nontrivial approximation, which implies a polynomial-time complexity. Thus, our recovery guarantees are also exponentially better compared to standard many-body tomography results.
\subsection{Results for shallow circuits}\label{sec:shallowciurcuits}

Let us be more explicit on how to leverage our results to also cover the outputs of short-depth circuits. To this end, let $G=(V,E)$ be a graph that models the interactions of a unitary circuit and suppose we implement $L$ layers of an unknown unitary circuit consisting of $1$ and $2$-qubit unitaries laid out according to $G$. That is, we have an unknown shallow circuit $\mathcal{U}$ of depth $L$ with respect to $G$. More precisely,
\begin{align}\label{equ:unitary_circuit}
    \mathcal{U}=\prod_{\ell\in [L]}\,\bigotimes_{e\in \mathcal{E}_{\ell}}\,\mathcal{U}_{\ell,e},
\end{align}
where each $\mathcal{E}_{\ell}\subset E$ are a subset of the edges such that any $e,e'\in \mathcal{E}_{\ell}$ do not share a vertex. Our goal is to show how to approximately learn the state
\begin{align*}
    \ket{\psi}=\mathcal{U}\ket{0}^{\otimes n}.
\end{align*}
The overall idea consists in finding a Gibbs state approximating $|\psi\rangle$ in Wasserstein distance. We will then find a differential structure for (approximations) of shallow circuits and then showing that the (approximation of the) output satisfies a TC inequality with respect to it. Thus, it suffices to control the relative entropy with this approximation to ensure a good approximation in Wasserstein distance.

Let us find the appropriate approximation. First, note that for $\beta_\epsilon=\log(\epsilon^{-1})$ and $H_0=-\sum_i^n Z_i$ where $Z_i$ denotes the Pauli matrix $Z=|0\rangle\langle 0|-|1\rangle\langle 1|$ acting on site $i$, we have that:
\begin{align}\label{equ:bound_entropy_pure}
D\Big(\mathcal{U}\ketbra{0}^{\otimes n}\mathcal{U}^\dagger\Big\|\mathcal{U}\,\frac{e^{-\beta_\epsilon H_0}}{\tr{e^{-\beta_\epsilon H_0}}}\,\mathcal{U}^\dagger\Big)&=D\Big(\ketbra{0}^{\otimes n}\Big\|\frac{e^{-\beta_\epsilon H_0}}{\tr{e^{-\beta_\epsilon H_0}}}\Big)\nonumber\\
&=nD\Big(\ketbra{0}\Big\|\frac{e^{\beta_\epsilon Z}}{\tr{e^{\beta_\epsilon Z}}}\Big)\nonumber\\
&=-n\beta_\epsilon \,\langle 0|Z|0\rangle\,+n\log\tr{e^{\beta_\epsilon Z}}\nonumber\\
&=n \log\big(1+e^{-2\beta_\epsilon }\big)\nonumber\\
&\leq ne^{-2\beta_\epsilon}\nonumber\\
&=\epsilon^2 n.
\end{align}

Thus, if the states $\mathcal{U}\frac{e^{-\beta_\epsilon H_0}}{\tr{e^{-\beta_\epsilon H_0}}}\mathcal{U}^\dagger$ satisfy a  transportation cost inequality with some constant $\alpha$, this would allow us to conclude that
\begin{align*}
    W_1\left(\mathcal{U}\ketbra{0}^{\otimes n}\mathcal{U}^\dagger,\,\mathcal{U}\frac{e^{-\beta_\epsilon H_0}}{\tr{e^{-\beta_\epsilon H_0}}}\mathcal{U}^\dagger\right)\leq\epsilon\, \sqrt{\frac{n}{2\alpha}}.
\end{align*}
Moreover, defining $H_{\mathcal{U}}=-\sum_i^n \mathcal{U}Z_i\mathcal{U}^\dagger$, we see that $\mathcal{U}\frac{e^{-\beta_\epsilon H_0}}{\tr{e^{-\beta_\epsilon H_0}}}\mathcal{U}^\dagger=\frac{e^{-\beta_\epsilon H_{\mathcal{U}}}}{\tr{e^{-\beta_\epsilon H_0}}}$.
As we know the geometry of $\mathcal{U}$, we can bound the support of $\mathcal{U}Z_i\mathcal{U}^\dagger$. Thus, we only need to find a suitable transportation cost inequality to see that this approximation fits into our framework.

Let us now find a suitable differential structure for the state $\frac{e^{-\beta_\epsilon H_0}}{\tr{e^{-\beta_\epsilon H_0}}}$. For simplicity, denote by $\tau_p=p\ketbra{0}+(1-p)\ketbra{1}$ and note that $\frac{e^{-\beta_\epsilon H_0}}{\tr{e^{-\beta_\epsilon H_0}}}=\tau_{p_\epsilon}^{\otimes n}$ with $p_\epsilon=\frac{e^{\beta_\epsilon}}{e^{\beta_\epsilon}+e^{-\beta_\epsilon}}$. Moreover, let 
\begin{align*}
a_i=I^{\otimes i-1}\otimes \ket{0}\bra{1}\otimes I^{\otimes n-i-1}
\end{align*}
be the anihilation operator acting on qubit $i$.  Defining $L_{i,0}=\lb p(1-p)\rb^{\frac{1}{4}}a_i$ and $L_{i,1}=L_{i,0}^\dagger$, we get a differential structure for $\tau_{p}^{\otimes n}$ with $\{L_{i,k},\omega_{i,k}\}$ for $i=1,\ldots,n$, $k=0,1$, $\omega_{i,0}=\frac{p}{1-p}$ and $\omega_{i,1}=\frac{1-p}{p}$. That this is indeed a differential structure follows from a simple computation.
One can readily check that the induced Lipschitz constant is given by:
\begin{align*}
\|O\|_{\operatorname{Lip},\nabla}^2&=\sum_{i=1}^n \lb \sqrt{\frac{p}{1-p}}+\sqrt{\frac{1-p}{p}}\rb(\|[O,L_{i,0}]\|^2+\|[O,L_{i,1}]\|^2)\\
&=\sum_{i=1}^n \|[O,a_i]\|^2+\|[O,a_i^\dagger]\|^2.
\end{align*}
Thus, we see that the Lipschitz constant takes a particularly simple form for this differential structure. 
Moreover, it is not difficult to see that $\{\mathcal{U}L_{i,k}\mathcal{U}^{\dagger},\omega_{i,k}\}$ provides a differential structure for the state $\mathcal{U}\tau_p^{\otimes n}\mathcal{U}^{\dagger}$. Indeed, it is easily checked that this new differential structure still gives eigenvectors of the modular operator. 
Importantly, the result of~\cite[Theorem 19]{beigi2020quantum} establish that the state $\tau_p^{\otimes n}$ satisfies a transportation cost inequality with constant $\frac{1}{2}$. 
Putting all these elements together we have:

\begin{theorem}[transportation cost for shallow circuits]\label{shallowcircuits}
Let $\mathcal{U}$ be an unknown depth $L$ quantum circuit on $n$ qubits defined on a graph $G=(V,E)$ and $\ket{\psi}=\mathcal{U}\ket{0}^{\otimes n}$.  Define for $\epsilon>0$ $H_{\mathcal{U}}=-\sum_i\mathcal{U}Z_i\mathcal{U}^\dagger$ and $\sigma(\mathcal{U},\epsilon)=\frac{e^{-\beta_\epsilon H_{\mathcal{U}}}}{\tr{e^{-\beta_\epsilon H_{\mathcal{U}}}}}$ with $\beta_\epsilon=\log(\epsilon^{-1})$. Then for any state $\rho$ and all observables $O$ we have:
\begin{align*}
|\tr{O(\ketbra{\psi}-\rho)}|\leq\|O\|_{\operatorname{Lip},\nabla}\left(\epsilon\sqrt{n}+\sqrt{D(\rho\|\sigma(\mathcal{U},\epsilon))}\right)
\end{align*}
with
\begin{align*}
\|O\|_{\operatorname{Lip},\nabla}^2=\sum_{i=1}^n \|[O,\mathcal{U}a_i\mathcal{U}^{\dagger}]\|^2+\|[O,\mathcal{U}a_i^\dagger \mathcal{U}^{\dagger}]\|^2,\quad a_i=I^{i-1}\otimes \ket{0}\bra{1}\otimes I^{n-i-1}
\end{align*}
\end{theorem}
\begin{proof}
We have:
\begin{align*}
    |\tr{O(\ketbra{\psi}-\rho)}|\leq |\tr{O(\ketbra{\psi}-\sigma(\mathcal{U}))}|+|\tr{O(\sigma(\mathcal{U})-\rho)}|.
\end{align*}
The claim then follows from the discussion above, as $\sigma(\mathcal{U})$ satisfies a transportation cost inequality with constant $\frac{1}{2}$.
\end{proof}

Of course the result above has the downside that the Lipschitz constant depends on the unknown circuit $\mathcal{U}$. However, as we can estimate the locality of each term $\mathcal{U}a_i \mathcal{U}^{\dagger}$, it is also possible to estimate the Lipschitz constant by controlling the overlap of the observable $O$ with each $\mathcal{U}a_i \mathcal{U}^{\dagger}$.

The result of Thm.~\ref{shallowcircuits} and the discussion preceding it also give us a method of efficiently learning the outputs of shallow circuits, as illustrated by the following proposition:

\begin{prop}[Learning the outputs of shallow circuits]
Let $\mathcal{U}$ be an unknown $n$-qubit quantum circuit with known locality structure as in Eq.~\eqref{equ:unitary_circuit} and $\ket{\psi}=\mathcal{U}\ket{0}^{\otimes n}$. Moreover, define $l_0$ as 
\begin{align*}
l_0=\max\limits_{1\leq i\leq n}\operatorname{supp}|\mathcal{U}Z_i\mathcal{U}^\dagger|.
\end{align*}
For some $\epsilon>0$ we have that
\begin{align}\label{equ:samplecomplexity_circuit_wasserstein}
\cO(\epsilon^{-2}4^{3l_0}\log^4(n4^{l_0}\epsilon^{-1})\log(4^{l_0}n\delta^{-1}))
\end{align}
samples of $\ket{\psi}$ suffice to learn a Gibbs state $\sigma(\mu)$ that satisfies 
\begin{align}
W_{1,\nabla}(\ketbra{\psi},\sigma(\mu))\leq \sqrt{\epsilon n}
\end{align}\label{equ:bound_wassser_circuit}
with probability of success at least $1-\delta$.
Along similar lines we have that
\begin{align}\label{equ:bound_trace_circuit}
\cO(\epsilon^{-2}4^{3l_0}n^2\log^4(n4^{l_0}\epsilon^{-1})\log(4^{l_0}n\delta^{-1}))
\end{align}
samples suffice to learn a state $\sigma(\lambda)$ such that
\begin{align*}
\|\ketbra{\psi}-\sigma(\lambda)\|_{\operatorname{tr}}\leq \sqrt{\epsilon}.
\end{align*}
\end{prop}
\begin{proof}
Let $\mathcal{E}_i$ be a basis of Pauli operators for matrices on the support of $\mathcal{U}Z_i\mathcal{U}^{\dagger}$. By our assumption on $l_0$, we know that for each $i$ we have that $|\mathcal{E}_i|\leq 4^{l_0}$. For simplicity we will assume that there are no Pauli words that are contained in two distinct $\mathcal{E}_i$ and we will enumerate all different Pauli words as $\{E_{i,j}\}$ for $1\leq i\leq n$ and $1\leq j\leq 4^{l_0}$ indicating the elements of the different $\mathcal{E}_{i}$.
Thus, there is a $\lambda\in \R^{m}$ with $m\leq n4^{l_0}$ and $\|\lambda\|_{\ell_\infty}\leq1$ such that
\begin{align*}
\sigma(\mathcal{U},\epsilon)\propto\operatorname{exp}\left(-\beta_{\epsilon}\sum\limits_{i=1}^n\sum\limits_{E_{i,j}\in\mathcal{E}_i}\lambda_{i,j}E_{i,j}\right).
\end{align*}
Let $\tilde{\epsilon}>0$ be given. It follows from Eq.~\eqref{equ:bound_entropy_pure} and Pinsker's inequality that picking $\epsilon=\tfrac{\tilde{\epsilon}}{n16}$ is sufficient to ensure that
\begin{align}\label{equ:trace_dist_pure_Gibbs}
\|\ketbra{\psi}-\sigma(\mathcal{U},\tfrac{\tilde{\epsilon}}{n4})\|_{\operatorname{tr}}\leq \frac{\tilde{\epsilon}}{2}.
\end{align}
Measuring 
\begin{align}\label{equ:tildeepsilon}
\cO(4^{l_0}\tilde{\epsilon}^{-2}\log(4^{l_0}n\delta^{-1}))
\end{align}
copies of $\ket{\psi}$ is sufficient to obtain estimates of $\tr{\ketbra{\psi}E_{i,j}}$ up to an $\tilde{\epsilon}/4$ error. By Eq.~\eqref{equ:trace_dist_pure_Gibbs} and a triangle inequality they are also $\tilde{\epsilon}/2$ away from $\tr{\sigma(\mathcal{U},\tfrac{\tilde{\epsilon}}{n4})E_{i,j}}$.

Thus, running the maximum entropy principle with these estimates and the basis of operators given by $\cup_{i=1}^n\mathcal{E}_i$ will yield us an estimate $\sigma(\mu)$ that satisfies:
\begin{align*}
D(\sigma(\mu)\|\sigma(\mathcal{U},\tfrac{\tilde{\epsilon}}{n}))\leq  \log\left(\frac{n}{4\tilde{\epsilon}}\right)\tilde{\epsilon}m\leq \log\left(\frac{n}{4\tilde{\epsilon}}\right)\tilde{\epsilon}4^{l_0}n.
\end{align*}
To obtain the estimate in Wasserstein distance in Eq.~\eqref{equ:samplecomplexity_circuit_wasserstein}, we can pick $\tilde{\epsilon}=\cO(\epsilon/(4^{l_0}\log^2(n4^{l_0}\epsilon^{-1})))$, as in this case
\begin{align*}
\log\left(\frac{n}{4\tilde{\epsilon}}\right)\tilde{\epsilon}m\leq \log\left(\frac{n}{4\tilde{\epsilon}}\right)\tilde{\epsilon}4^{l_0}n=\epsilon\frac{\log\left(n\log^2(n4^{l_0}\epsilon^{-1})\epsilon^{-1}\right)}{\log^2(n4^{l_0}\epsilon^{-1})}n\leq \epsilon n.
\end{align*}
The claim then follows from the results of Thm.~\ref{shallowcircuits} and substituing $\tilde{\epsilon}$ into Eq.~\eqref{equ:tildeepsilon}.

For the statement on the sample complexity for the trace distance, we may pick $\tilde{\epsilon}=\cO(\epsilon/(n4^{l_0}\log^2(n4^{l_0}\epsilon^{-1})))$, which yields the statement after applying Pinsker's inequality.
\end{proof}

This shows that shallow circuits can be learned efficiently as long as $l_0=\cO(\log(n))$. However, it is not immediately obvious how to also estimate the expectation values of the Gibbs states required to run the maximum entropy algorithm. Thus, at least with the methods presented here, the postprocessing takes exponential time in the number of qubits.

\paragraph{Ground states of gapped systems:} in light of our results for shallow circuits, it is natural to ask to what extent our framework can be extended to ground states of gapped Hamiltonians, specially in 1D. Thus, let us briefly comment on the technical barriers on the way of such statements. First, notice that the statement of Thm~\ref{shallowcircuits} required inverse temperatures scaling logarithmically in system size for the approximation of the ground state. Most known TC inequalities have an exponential scaling with the inverse temperature and, thus, TC at this inverse temperature the savings compared to Pinsker are not quadratic, hindering a straightforward extension to gapped systems. There are some nontrivial examples of ground states satisfying a TC inequality with the constant depending inverse linearly with the temperature, like graph states~\cite{temme_hypercontractivity_2014}. But, as they can also be prepared from a constant depth quantum circuit almost by definition, they fall into the assumptions of the previous statement.

However, the results of~\cite{dalzell_locally_2019} assert that $k$-local density matrices of ground states of gapped Hamiltonians in $1D$ can be approximated by constant depth circuits, giving evidence that our framework should also extend to such states. And to get there, a technical obstacle has to be overcome in the proof of Thm.~\ref{shallowcircuits}. Essentially we need to show that local reduced density matrices of the ground state are already well-approximated at inverse temperature $\log(\epsilon^{-1})$. With such a statement we could show that we can still approximate the expectation values of $E_i$ at inverse temperature $\log(\epsilon^{-1})$ from samples from the ground state $\ket{\psi}$.

\section{Summary of known strong convexity constants}\label{sec:strong_convexity}
As we see in the statement of Thm.~\ref{thm:comcomplexitygeneral} and \ref{thm:quick_dirty}, having a bound on the strong convexity constant $L^{-1}$ can give a quadratic improvement in the sample complexity in terms of the error $\epsilon$. Here we will briefly summarize for which cases estimates on $L$ are known in the literature for the classes of states we considered here.

\paragraph{General many-body quantum:} first, we should mention the results of~\cite{2004.07266}. There the authors show bounds on $L^{-1}$ for arbitrary many-body Hamiltonians and temperatures that scales linearly in $m$. Thus, although certainly nontrivial, these bounds do not improve the sample complexity for the regimes we are interested in this work, namely that of logarithmic sample complexity in system size. To obtain improvements in this regime, $L^{-1}$ needs to be at most polylogarithmic in system size.

In the case of high-temperature Gibbs states, the recent work of~\cite{tang_hamiltonian} shows that this is indeed the case. I.e., in their Corollary 4.4 they show that for $\beta=\cO( k^{-8})$, where $k$ is the locality of the Hamiltonian, we indeed have $L^{-1}=\cO(\beta^{-2})$. It should also be noted that their results do not hold only for geometrically local Hamiltonians, but rather any Hamiltonian such that each term acts on at most $k$ qudits. This implies that for the high temperature regime, for which we also have the TC inequality in Thm.~\ref{thm:modifiedmlsi}, the improved sample complexity yielded by our methods holds. Note, however, that there is a slight mismatch between the inverse temperature range for which the two results hold: for the strong convexity we need $\beta=\cO(k^{-8})$, whereas for TC $\beta=\cO(k^{-1})$ suffices.

\paragraph{Commuting Hamiltonians:} as we will prove later in Prop.~\ref{prop:lowerbound}, in the case of commuting Hamiltonians we have that $L^{-1}=\cO(e^{\beta}\beta^{-2})$. Thus, for any constant inverse temperature $\beta>0$ we have an improved sample complexity. However, in order to analyse ground states in $1D$, our current proof techniques still require an inverse temperature scaling logarithmically in system size, so for such states we do not obtain improvements through strong convexity. We plan to address this gap in future work.

Besides the cases mentioned above, we also considered the case of local circuits in this work. For those there are no nontrivial estimates on $L$ available to the best of our knowledge.

\section{Regimes of efficient postprocessing}\label{sec:regimesefficient}
The only question we have still to answer is how to perform the postprocessing efficiently, namely how the parameter $C_{\mathcal{E}}$ appearing in Theorem \ref{thm:quick_dirty} scales and how we obtain the bounds in Table~\ref{tab:performance}. 

There have been many recent breakthroughs on approximating quantum Gibbs states efficiently on classical computers~\cite{kliesch_locality_2014,PhysRevLett.124.220601,harrow2020classical,kuwahara2020gaussian,kuwahara_improved_2021}. The gradient descent method only requires us to estimate the gradient of the partition function for a Gibbs state at each iteration. Thus, any algorithm that can approximate the log-partition function $\mathcal{Z}(\mu)$ efficiently or approximate $e(\lambda)$ suffices for our purposes. 

For Gibbs states on a lattice, the methods of~\cite{PhysRevLett.124.220601} yield that we can perform such approximations on a classical computer in time polynomial in $n$ for temperatures $\beta<\beta_c=1/(k8e^3)$, where $k$ is the locality of the Hamiltonian, and $\epsilon$ inverse polynomial in system size. Thus, we conclude that for this temperature range, which coincides with the range for which the results of Thm.~\ref{thm:modifiedmlsi} hold, $C_{\mathcal{E}}$ is polynomial in system size and we can also obtain efficient classical postprocessing.

For the case of Gibbs states of $1$D systems, to the best of our knowledge, the best results available are those of ~\cite{kuwahara_improved_2021}. They show how to efficiently obtain efficient tensor network approximations for $1$D Gibbs states for $\beta=o(\log(n))$. As such tensor networks can be contracted efficiently as well, this gives an efficient classical algorithm to compute local expectation values of such states, which suffices for our purposes. Thus, the results of~\cite{PhysRevLett.124.220601,kuwahara_improved_2021} ensure that for the systems considered in Thm.~\ref{thm:modifiedmlsi} we can also perform the postprocessing efficiently on a classical computer. 

It is also interesting to consider what happens if we have access to a quantum computer to estimate the gradient. We will discuss the implications of this in the next section for the case of commuting Hamiltonians.

Finally, for local quantum circuits we are not aware at this stage of any method that could yield a better postprocessing complexity than computing the partition function explicitly. This would yield a postprocessing that is exponential in the system size, as it requires diagonalizing an exponentially large matrix.

\section{A complete picture: commuting Gibbs states}\label{sec:completepicturecommuting}
In this section, we discuss two classes of states for which the strongest version of our theorems holds, namely that of commuting $1$D Gibbs states, and the one of high-temperature commuting Gibbs states on arbitrary hypergraphs. We already discussed that they satisfy transportation cost inequalities in Thm.~\ref{thm:modifiedmlsi}. Thus, the last missing ingredient to obtain an optimal performance is to show that the partition function is indeed strongly convex. More precisely, we will now establish that, for these classes of states, both the upper and lower bounds on the Hessian of the log partition function are order $1$. In addition to that, with access to a quantum computer, it is possible to perform the post-processing in time $\tcO(m)$. As writing down the vector $\lambda$ takes $\Omega(m)$ time, we conclude that the postprocessing can be performed in a time comparable to writing down a solution to the problem. Thus, our procedure is essentially optimal.

Also in the setting of commuting Gibbs states, it is worth noting that after the completion of this work, we became aware of~\cite{note_anurag}, which gives another method to learn the Gibbs state and its Hamiltonian that neither involves the maximum entropy method nor requires strong convexity guarantees. Their algorithm works by learning local reduced density matrices and showing that the parameters $\lambda$ of the Hamiltonian of a commuting Gibbs state can be efficiently estimated from that. In principle, obtaining a bound on $\lambda$ also suffices for our purposes and we could alternatively use their methods to bypass having to solve the maximum entropy problem for such states. In particular, this means that the postprocessing with their methods could be performed even for temperatures at which the partition function cannot be estimated efficiently but we still have access to samples from the state.
However, as we ultimately are interested in the regime in which TC inequalities hold, which corresponds to the high-temperature regime, we do not further comment on their results.

In this section, we consider a hypergraph $G=(V,E)$ and assume that there is a maximum radius $r_0\in\mathbb{N}$ such that, for any hyperedge $A\in E$, there exists a vertex $v\in V$ such that the ball $B(v,r_0) $ centered at $v$ and of radius $r_0$ includes $A$. 
In what follows, we also denote by $S(v,r)$ the sphere of radius $r$ centered at vertex $v\in V$, and define for all $r\in \mathbb{N}$:
\begin{align*}
    B(r):=\max_{v\in V}|B(v,r)|\,,\qquad \qquad S(r):=\max_{v\in V}|S(v,r)|\,.
\end{align*}
Next, we consider a Gibbs state $\sigma(\lambda)$ on $\cH_V:=\bigotimes_{v\in V}\cH_v$, where $\operatorname{dim}(\cH_v)=d$ for all $v\in V$. More precisely, $\sigma(\lambda):=e^{-\beta H(\lambda)}/\mathcal{Z}(\lambda)$, where with a slight abuse of notations
\begin{align}\label{eq:hamiltonian}
    H(\lambda)=\sum_{i=1}^m\lambda_i E_i=\sum_{A\in E}\,\sum_{\alpha_A\in [d^2]^{|A|}}\alpha_A\,E_{\alpha_A},
\end{align}
with $\|E_i\|_\infty=1$ for all $i\in\{1,\cdots, m\}$, $[E_i,E_j]=0$ for all $i,j\in \{1,\cdots, m\}$, and where the sets $\{1,\cdots,m\}$ and $\{\alpha_A\}_{A\in E,\alpha_A\in [d^2]^{A}}$ are in bijection. Note that $B(r_0)$ bounds the maximal locality of the Hamiltonian.

We also denote by $\sigma_{A}(\lambda)$ the Gibbs state corresponding to the restriction of $H$ onto the region $A$, i.e. 
 \begin{align*}
     \sigma_A(\lambda):=\frac{e^{-\beta H_A(\lambda)}}{\tr{e^{-\beta H_A(\lambda)}}}\,,~~~~~\text{ where }~~~~H_A(\lambda):=\underset{i|\operatorname{supp}(E_i)\cap A\ne 0}{\sum} \lambda_i E_i\,.
 \end{align*}
Note that, in general, $\sigma_A(\lambda)\ne \Tr_{A^c}(\sigma(\lambda))$. 

\subsection{Upper bound on Hessian for commuting Gibbs states with decay of correlations}

In this section, we prove tightened strong convexity constants for the log partition function in the case when the Gibbs state arises from a local commuting Hamiltonian at high-temperature. In fact, the upper constant found in Proposition \ref{upperconstant} can be tightened into a system size-independent one under the condition of exponential decay of the correlations in the Gibbs state. 
Several notions of exponential decay of correlations exist in the literature~\cite{capel2020modified}. Here, we will say that a Gibbs state $\sigma$ has correlation length $\xi$ if for all observables $O_A,O_B$ supported on non-overlapping regions $A$ and $B$ respectively, we have that:
\begin{align*}
    |\tr{\sigma\, O_A\otimes O_B}-\tr{\sigma\, O_A}\tr{\sigma \,O_B}|\leq c\,\|O_A\|_\infty\,\|O_B\|_\infty e^{-\xi \operatorname{dist}(A,B)}
\end{align*}
for some constant $c>0$. In the classical setting, this condition is known to hold at any inverse temperature for $1$D systems, and below a critical inverse temperature $\beta_c$ that depends on the locality of the Hamiltonian when $D\ge 2$ \cite{dobrushin1987completely}. The same bound also holds in the quantum case for $1$D systems \cite{araki1969gibbs}, or above a critical temperature on regular lattices when $D\ge 2$ \cite{harrow2020classical,kliesch_locality_2014}. Using these bounds, we obtain the following improvement of Proposition \ref{upperconstant} which shows that for this class of states $U=\cO(1)$.

\begin{lem}[Strengthened upper strong convexity at high-temperature and for 1D systems]\label{lem:covariancematrix_decay}
For each $\mu\in B_{\ell_\infty}(0,1)$, let $\sigma(\mu)$ be a Gibbs state at inverse temperature $\beta<\beta_c$ corresponding to the Hamiltonian defined on the hypergraph $G=(V,E)$ in \eqref{eq:hamiltonian}. Then
\begin{align*}
    \nabla^2\,\log(\mathcal{Z}(\mu))\le \Big(c\,\beta^2\,B(r_0)\,B(2r_0)\,d^{2B(r_0)}  \sum\limits_{r=1}^\infty e^{-\xi r}S(r)\Big)\,I\,,
\end{align*}
where $\xi$ is the correlation length of the state. Moreover, this result holds for all $\beta>0$ in $1$D.
\end{lem}
\begin{proof}
Let us first use the assumption of commutativity to simplify the expression for the Hessian. We find for all $\alpha_A\in [d^2]^{|A|},\alpha_B'\in [d^2]^{|B|}$, $A,B\in E$,
\begin{align}\label{eq:hesscomm}
    (\nabla^2\,\log(\mathcal{Z}(\mu)))_{\alpha_A\alpha_B'}=\beta^2\,\Tr\big[\sigma(\mu)\,(E_{\alpha_A}-\tr{\sigma(\mu)\,E_{\alpha_A}})\,(E_{\alpha_B'}-\tr{\sigma(\mu)\,E_{\alpha_B'}})\big]\,.
\end{align}
The rest follows similarly to Proposition \ref{upperconstant} from Gershgorin's circle theorem together with the decay of correlations arising at $\beta<\beta_c$: for all $\alpha_A$,
\begin{align*}
    \sum\limits_{\alpha_B'\not= \alpha_A} \left|\left(\nabla^2\,\log(\mathcal{Z}(\mu))\right)_{\alpha_A\alpha_B'}\right|\le c\,\beta^2\,\sum\limits_{\alpha_B'\not=\alpha_A} e^{-\xi \operatorname{dist}(A,B)}\,,
\end{align*}
where we also used that the basis operators $\{E_i\}$ have operator norm $1$. The claim then follows by bounding the number of basis operators whose support is at a distance $r$ of $A$ for each $r\in\mathbb{N}$: the latter is bounded by the product of (i) the number of vertices in $A$; (ii) the number of vertices at a distance $r$ of a given vertex; (iii) the number of hyperedges containing a given vertex; and (iv) the number of basis operators corresponding to a given hyperedge. A simple estimate of each of these quantities gives the bound $B(r_0)\,S(r)\,B(2r_0)\,d^{2B(r_0)}$. Therefore:
\begin{align*}
  \sum\limits_{\alpha_B'\not= \alpha_A} \left|\left(\nabla^2\,\log(\mathcal{Z}(\mu))\right)_{\alpha_A\alpha_B'}\right|\le c\,\beta^2\,B(r_0)\,B(2r_0)\,d^{2B(r_0)}  \sum\limits_{r=1}^\infty e^{-\xi r}S(r)\,.
\end{align*}
\end{proof}
Note that for $D$-regular graphs and $r_0=\cO(1)$ we have $B(r_0)\,B(2r_0)\,d^{2B(r_0)}=\cO(1)$ and $S(r)=\cO(r^{D-1})$, giving a scaling of $ \nabla^2\,\log(\mathcal{Z}(\mu))=\cO(\beta^2\xi^D)$.

\subsection{Lower bound on Hessian for commuting Gibbs states}

Whenever the Gibbs state is assumed to be commuting, the lower strong convexity constant $L$ can also be made independent of $m$, by a direct generalization of the classical argument as found in \cite{vuffray2016interaction}[Lemma 7] or \cite{montanari2015computational} (see also \cite{2004.07266}).

Before we state and prove our result, let us introduce a few useful concepts: given a full-rank quantum state $\sigma$, we denote the weighted $2$-norm of an observable $Y$ as
\begin{align*}
\|Y\|_{2,\sigma}:=\tr{\big|\sigma^{\frac{1}{4}} Y\,\sigma^{\frac{1}{4}}\big|^2}^{1/2}\,.
\end{align*}
and refer to the corresponding non-commutative weighted $L_2$ space with inner product $\langle X,Y\rangle_{\sigma}:=\tr{\sigma^{\frac{1}{2}}X^\dagger\sigma^{\frac{1}{2}}Y}$ as $L_2(\sigma)$.
The Petz recovery map corresponding to a subregion $A\subset V$ with respect to $\sigma(\mu)$ is defined as 
\begin{align*}
\mathcal{R}_{A,\sigma(\mu)}(X)=\Tr_A(\sigma(\mu))^{-1/2}\Tr_A(\sigma(\mu)^{1/2}X\sigma(\mu)^{1/2})\Tr_A(\sigma(\mu))^{-1/2}\,,
\end{align*}
We will also need the notion of a conditional expectation $\mathbb{E}_A$ with respect to the state $\sigma(\mu)$ into the region $A\subset \Lambda$ (see \cite{kastoryano2016quantum,bardet2020approximate} for more details). For instance, one can choose  $\mathbb{E}_A:=\lim_{n\to\infty }\mathcal{R}_{A,\sigma(\mu)}^n$, where $\mathcal{R}_{A,\sigma(\mu)}$ is the Petz recovery map of $\sigma(\mu)$. In other words, the map $\mathbb{E}_A$ is a completely positive, unital map that projects onto the algebra $\mathcal{N}_A$ of fixed points of $\mathcal{R}_{A,\sigma(\mu)}$. This algebra is known to be expressed as the commutant \cite{bardet2020approximate}
\begin{align*}
    \mathcal{N}_A:=\{\sigma(\mu)^{it}\mathcal{B}(\cH_A)\sigma(\mu)^{-it};\,t\in\mathbb{R}\}'\,.
\end{align*}
Moreover, the maps $\mathbb{E}_A$ commute with the modular operator $\Delta_{\sigma(\mu)}(.):=\sigma(\mu)(.)\sigma(\mu)^{-1}$, and for any $X_A,Z_A\in \cN_A$ and all $Y\in \cB(\cH_V)$, $\mathbb{E}_A[X_AYZ_A]=X_A\mathbb{E}_A[Y]Z_A$. The commutativity condition for $H(\mu)$ implies frustration freeness of the family of conditional expectations $\{\mathbb{E}_A\}_{A\in E}$: for any $X\in L_2(\sigma(\mu))$, $\|\mathbb{E}_A(X)\|_{2,\sigma(\mu)}^2+\|(\operatorname{id}-\mathbb{E}_A)(X)\|_{2,\sigma(\mu)}^2=\|X\|_{L_2(\sigma(\mu))}^2$.

The next technical lemma is essential in the derivation of our strong convexity lower bound. With a slight abuse of notations, we use the simplified notations $\sigma_{x}(\lambda):=\sigma_{\{x\}}(\lambda)$, $H_x(\lambda):=H_{\{x\}}(\lambda)$ and so on.
\begin{lem}\label{technlemmacontract}
Let $H=\sum_j\,\mu_j\,E_j$ be a local commuting Hamiltonian on the hypergraph $G=(V,E)$ defined in \eqref{eq:hamiltonian}, each local operator $E_j$ is further assumed to be traceless. The following holds for any $x\in V$:
\begin{align*}
    c(x,\beta):=\max_{\mu\in B_{\ell_\infty}(0,1)}\frac{\|\mathcal{R}_{x,\sigma_x(\mu)}(H_x)-\tr{\sigma_x(\mu)H_x}I\|_{2,\sigma_x(\mu)}}{\|H_x-\tr{\sigma_x(\mu) H_x}I\|_{2,\sigma_x(\mu)}}<1\,.
\end{align*}
\end{lem}

\begin{proof}
 We first prove that $X= \mathcal{R}_{x,\sigma_x(\mu)}(X)$ is equivalent to $\|\mathcal{R}_{x,\sigma_x(\mu)}(X)\|_{2,\sigma_x(\mu)}=\|X\|_{2,\sigma_x(\mu)}$. One direction trivially holds. For the opposite direction, assume that $X\ne \mathcal{R}_{x,\sigma_x(\mu)}(X)$. This means that $X=Y+Z$, with $Y,Z$ two operators that are orthogonal in $L_2(\sigma_x(\mu))$, with $\mathcal{R}_{x,\sigma_x(\mu)}(Y)=Y$ and $Z\ne 0$. Now, since $\mathcal{R}_{x,\sigma_x(\mu)}$ is self-adjoint and unital, it strictly contracts elements in the orthogonal of its fixed points and we have
\begin{align*}
    \|\mathcal{R}_{x,\sigma_x(\mu)}(X)\|_{2,\sigma_x(\mu)}^2&=\|Y\|_{2,\sigma_x(\mu)}^2+\|\mathcal{R}_{x,\sigma_x(\mu)}(Z)\|_{2,\sigma_x(\mu)}^2\\
    &<\|Y\|_{2,\sigma_x(\mu)}^2+\|Z\|_{2,\sigma_x(\mu)}^2\\
    &=\|X\|_{2,\sigma_x(\mu)}^2\,,
\end{align*}
which contradicts the condition of equality of the norms. Now, since the map $\mathcal{R}_{x,\sigma_x(\mu)}$ is unital, it suffices to prove that $\mathcal{R}_{x,\sigma_x(\mu)}(H_x)\ne H_x$, or equivalently $\mathbb{E}_x(H_x)\ne H_x$, in order to conclude. Let us assume instead that equality holds. This means that, for all observables $A_x$ supported on site $x$, and all $t\in\mathbb{R}$:
\begin{align*}
    [H_x,\sigma(\mu)^{it}A_x\sigma(\mu)^{-it}]=\sigma(\mu)^{it}[H_x,A_x]\sigma(\mu)^{-it}=0\,~~~~\Rightarrow ~~~~H_x=\frac{I_x}{d}\otimes \Tr_{x}(H_x)\,. 
\end{align*}
However, this contradicts the fact that $H_x$ is traceless on site $x$. Therefore $\mathcal{R}_{x,\sigma_x(\mu)}(H_x)\ne H_x$ and the proof follows.
\end{proof}
We are ready to prove our strong convexity lower bound.

\begin{prop}[Strengthened lower strong convexity constant for commuting Hamiltonians]\label{prop:lowerbound}
For each $\mu\in B_{\ell_\infty}(0,1)$, let $\sigma(\mu)$ be a Gibbs state at inverse temperature $\beta$ corresponding to the commuting Hamiltonian $H(\mu)=\sum_j\,\mu_j\,E_j$ on the hypergraph $G=(V,E)$ defined in \eqref{eq:hamiltonian}, where $\tr{E_iE_j}=0$ for all $i\not=j$ and each local operator $E_j$ is traceless on its support. Then the Hessian of the log-partition function is lower bounded by
\begin{align*}
    \nabla^2\,\log\,\mathcal{Z}(\mu)\ge 
    \beta^2\,e^{-\beta (B(2r_0)+2B(4r_0))}d^{-B(2r_0)}\,(1-c(\beta)^2)\,I\,,
\end{align*}
where $c(\beta):=\max_{v\in V}c(v,\beta)$.
\end{prop}
\begin{proof}
We first use the assumption of commutativity in order to simplify the expression for the Hessian. As in \eqref{eq:hesscomm}, we find
\begin{align*}
    (\nabla^2\,\log(\mathcal{Z}(\mu)))_{ij}=\beta^2\,\Tr\big[\sigma(\mu)\,(E_i-\tr{\sigma(\mu)\,E_i})\,(E_j-\tr{\sigma(\mu)\,E_j})\big]\,.
\end{align*}
Therefore, for any linear combination $H\equiv H(\lambda)=\sum_{j}\,\lambda_j\,E_j$ of the basis vectors, we have
\begin{align*}
    \sum_{ij}\,\lambda_i\lambda_j\,(\nabla^2\log(\mathcal{Z}(\mu)))_{ij}=\beta^2\operatorname{Var}_{\sigma(\mu)}(H)\,,
\end{align*}
with $\operatorname{Var}_{\sigma(\mu)}(X):=\|X-\tr{\sigma(\mu) X}\|_{2,\sigma(\mu)}^2$.
It is thus sufficient to lower-bound the latter. For this, we choose a subregion $A\subseteq V$ such that any basis element $E_i$ has support intersecting a unique vertex in $A$. We lower-bound the variance by
\begin{align}
    \|H-\tr{\sigma(\mu)H}I\|_{2,\sigma(\mu)}^2\ge \|(\id-\mathbb{E}_A)(H-\tr{\sigma(\mu)H}I)\|_{2,\sigma(\mu)}^2= \|H-\mathbb{E}_A[H]\|_{2,\sigma(\mu)}^2\,,\label{variancetobound}
\end{align}
where the first inequality follows by the $L_2(\sigma(\mu))$ contractivity of $\id-\mathbb{E}_A$. Now, the weighted norm can be further simplified into a sum of local weighted norms as follows:  first, for any two $E_i,E_j$ whose supports intersect with a different vertex of $A$, we have 
\begin{align}
    \langle \mathbb{E}_A[E_i],\,\mathbb{E}_A[E_j]\rangle_{\sigma(\mu)}&=\tr{\sigma(\mu)^{1/2}\mathbb{E}_A[E_i]\sigma(\mu)^{1/2}\mathbb{E}_A[E_j] }\nonumber\\
    &=\tr{\sigma(\mu)\Delta_{\sigma(\mu)}^{-1/2}\circ \mathbb{E}_A[E_i] \mathbb{E}_A[E_j]}\nonumber\\
    &=\tr{\sigma(\mu)\mathbb{E}_A[E_i] \mathbb{E}_A[E_j]}\label{eqoverlap}
\end{align}
where in the third line we used the commutation of the modular operator $\Delta_{\sigma(\mu)}(X):=\sigma(\mu)X\sigma(\mu)^{-1}$ with $\mathbb{E}_A$ together with the commutativity of $\sigma(\mu)$ and $E_i$. Now, denoting $\operatorname{supp}(E_i)\cap A=\{x\}$ and $\operatorname{supp}( E_j)\cap A=\{y\}$, we show that
\begin{equation}
\label{eaec}
    \mathbb{E}_A[E_i]=\mathbb{E}_x[E_i]~~~~~\text{ and }~~~~~ \mathbb{E}_A[E_j]=\mathbb{E}_y[E_j]\,.
\end{equation}
In order to prove these two identities, we simply need to prove for instance that $\mathbb{E}_x[E_i]$ belongs to the image algebra $\mathcal{N}_A$ of $\mathbb{E}_A$ since $\cN_A\subseteq\cN_x$ by definition. Hence, it is enough to show that $\mathbb{E}_x[E_i]$ commutes with operators of the form $\sigma(\mu)^{it}X_A\sigma(\mu)^{-it}$ for any $t\in\mathbb{R}$ and any $X_A\in\mathcal{B}(\mathcal{H}_A)$. This claim follows from
\begin{align*}
    \mathbb{E}_x[E_i]\sigma(\mu)^{it}X_A\sigma(\mu)^{-it}&=\sigma(\mu)^{it}\sigma(\mu)^{-it}\mathbb{E}_x[E_i]\sigma(\mu)^{it}X_A\sigma(\mu)^{-it}\\
    &=\sigma(\mu)^{it}\mathbb{E}_x[\sigma(\mu)^{-it}E_i\sigma(\mu)^{it}]X_A\sigma(\mu)^{-it}\\
    &=\sigma(\mu)^{it}\mathbb{E}_x[E_i]X_A\sigma(\mu)^{-it}\\
    &=\sigma(\mu)^{it}X_A\mathbb{E}_x[E_i]\sigma(\mu)^{-it}\\
    &=\sigma(\mu)^{it}X_A\sigma(\mu)^{-it}\mathbb{E}_x[E_i]\,.
\end{align*}
where the fourth line follows from the fact that the support of $\mathbb{E}_x[E_i]$ does not intersect $A\backslash\{i\}$, together with the fact that $\mathbb{E}_x[E_i]$ is locally proportional to $I$ on site $x$, by definition of $\mathcal{N}_x$. Therefore, using (\ref{eaec}) into (\ref{eqoverlap}), we get
\begin{align*}
    \langle \mathbb{E}_A[E_i],\,\mathbb{E}_A[E_j]\rangle_{\sigma(\mu)}&=\tr{\sigma(\mu)\mathbb{E}_x[E_i] \mathbb{E}_y[E_j]}\\
    &=\tr{\sigma(\mu)\mathbb{E}_x\mathbb{E}_y[E_iE_j] }\\
    &=\tr{\sigma(\mu)E_iE_j}\\
    &=\langle E_i,\,E_j\rangle_{\sigma(\mu)}\,,
\end{align*}
where in the second line we used that $\mathbb{E}_x[E_i]$ is a fixed point of $\mathbb{E}_y$, and then that $E_j$ is a fixed point of $\mathbb{E}_x$, by the support conditions of $E_i$ and $E_j$. Therefore, the variance on the right-hand side of (\ref{variancetobound}) can be simplified as
\begin{align*}
    \|H-\mathbb{E}_A[H]\|_{2,\sigma(\mu)}^2&=\sum_{x\in A}\,\big\|\sum_{j|\,\operatorname{supp}(E_j)\ni x}\lambda_j\,(E_j-\mathbb{E}_x[E_j])\big\|_{2,\sigma(\mu)}^2\\
    &=\sum_{x\in A}\,\big\|H_x-\mathbb{E}_x[H_x]\big\|_{2,\sigma(\mu)}^2
\end{align*}
where we recall that $H_x:=\sum_{j|\,\operatorname{supp}(E_j)\ni x}\lambda_j\,E_j$. 
Now, for any $x\in V$, we denote $x\partial:=\{A\in E|\,x\in A\}$ and decompose the Hamiltonian $H(\mu)$ as
\begin{align}
    H(\mu)=H_x(\mu)+K_x^0(\mu)+K_x^1(\mu)\,,\quad \text{ where }\quad \left\{\begin{aligned}
    &\operatorname{supp}(K_x^1(\mu))\cap x\partial =\emptyset\\
    &\operatorname{supp}(K_x^0(\mu))\cap x\partial \ne \emptyset
    \end{aligned}\right.\,.
\end{align}
Clearly, 
\begin{align}\label{ineqstate}
    \sigma(\mu)\ge \sigma_x(\mu)\,e^{-2\beta\,\|K_x^0(\mu)\|_\infty}\,\frac{e^{-\beta K_x^1(\mu)}}{\operatorname{tr}_{x\partial^c}\big[e^{-\beta K_x^1(\mu)}\big]}\,.
\end{align}
Now, 
\begin{align*}
    \|H_x-\mathbb{E}_x[H_x]&\|_{2,\sigma(\mu)}^2= \tr{\sigma(\mu)(H_x-\mathbb{E}_x[H_x])^2}\\
    &\ge \,e^{-2\beta\,\|K_x^0(\mu)\|_\infty}\,\tr{\sigma_x(\mu)\,\frac{e^{-\beta K_x^1(\mu)}}{\operatorname{tr}_{x\partial^c}\big[e^{-\beta K_x^1(\mu)}\big]}(H_x-\mathbb{E}_x[H_x])^2}\,\\
    &=e^{-2\beta\,\|K_x^0(\mu)\|_\infty}\,\|H_x-\mathbb{E}_x[H_x]\|_{2,\sigma_x(\mu)}^2\\
    &=e^{-2\beta\,\|K_x^0(\mu)\|_\infty}\big(\|H_x-\tr{\sigma_x(\mu)H_x} I\|_{2,\sigma_x(\mu)}^2-\|\mathbb{E}_x[H_x]-\tr{\sigma_x(\mu)H_x} I\|_{2,\sigma_x(\mu)}^2\big)\\
    &\ge e^{-2\beta\,\|K_x^0(\mu)\|_\infty}\big(\|H_x-\tr{\sigma_x(\mu)H_x} I\|_{2,\sigma_x(\mu)}^2-\|\mathcal{R}_{x,\sigma(\mu)}[H_x]-\tr{\sigma_x(\mu)H_x} I\|_{2,\sigma_x(\mu)}^2\big)\\
&    \ge (1-c(\beta)^2)\,e^{-2\beta\,\|K_x^0(\mu)\|_\infty}\,\|H_x-\tr{\sigma_x(\mu)H_x} I\|^2_{2,\sigma_x(\mu)}\,.
\end{align*}
The first and second identities above follow once again from the commutativity of the Hamiltonian similarly to \eqref{eqoverlap}, where for the second one we also use the disjointness of $x\partial$ and $\operatorname{supp}(K_x^1(\mu))$. The first inequality follows from \eqref{ineqstate}. The third identity is a consequence of the fact that $\mathbb{E}_x$ is a projection with respect to $L_2(\sigma_x(\mu))$. The second inequality follows from $\|\mathcal{R}_{x,\sigma(\mu)}[X]\|_{2,\sigma(\mu)}\ge \|\mathbb{E}_x[X]\|_{2,\sigma(\mu)}$ for all $X$ (see Proposition 10 in \cite{kastoryano2016quantum}). The last inequality is a consequence of Lemma \ref{technlemmacontract}. Finally, we further bound the weighted $L_2$ norm on the last line of the above inequality in terms of the Schatten $2$ norm to get
\begin{align}
    \operatorname{Var}_{\sigma(\mu)}(H)&\ge (1-c(\beta)^2)\,\min_{x\in V}\,e^{-2\beta\|K_x^0(\mu)\|_\infty}\,\lambda_{\min}(\sigma_x(\mu)) \sum_{x\in A}\,\|H_x-\tr{\sigma_x(\mu)\,H_x}\|_2^2\\
    &\ge(1-c(\beta)^2)\min_{x\in V}\,e^{-2\beta\|K_x^0(\mu)\|_\infty}\,\lambda_{\min}(\sigma_x(\mu))\,\sum_i\lambda_i^2\,,
\end{align}
where $\lambda_{\min}(\sigma_x(\mu))$ denotes the minimum eigenvalue of $\sigma_x(\mu)$. The result follows from the simple estimates $\|K_x^0(\mu)\|_\infty\le B(4r_0)$ and $\lambda_{\min}(\sigma_x(\mu))\ge e^{-\beta B(2r_0)}d^{-B(2r_0)}$.

\end{proof}

\subsection{Summary for 1D or high-temperature commuting}

Now we genuinely have all the elements in place to essentially give the last word on estimating Lipschitz observables for Gibbs states of nearest neighbor $1$D or high-temperature commuting Gibbs states.
\begin{theorem}
For each $\mu\in B_{\ell_\infty}(0,1)$, let $\sigma(\mu)$ be a Gibbs state at inverse temperature $\beta$ corresponding to the commuting Hamiltonian $H(\mu)=\sum_{j=1}^m\,\mu_j\,E_j$ on the hypergraph $G=(V,E)$ defined in \eqref{eq:hamiltonian}, where $\tr{E_iE_j}=0$ for all $i\not=j$, each local operator $E_j$ is traceless on its support and acts on a constant number of qubits and $m=\cO(n)$.
Moreover, assume that $\sigma(\lambda)$ satisfies the conditions of Thm.~\ref{thm:modifiedmlsi}. Then $\cO(\log(n)\epsilon^{-2})$ samples of $\sigma(\lambda)$ suffice to obtain a $\mu\in B_{\ell_\infty}(0,1)$ satisfying
\begin{align*}
\tr{O(\sigma(\lambda)-\sigma(\mu))}=\cO(\epsilon\sqrt{n}\|O\|_{\operatorname{Lip},\nabla}).
\end{align*}
with probability at least $1-p$.
Moreover, we can find $\mu$ in $\cO(\textrm{poly}(n,\epsilon^{-1}))$ time on a classical computer. With access to a quantum computer, the postprocessing can be performed in $\tcO(n\epsilon^{-2})$ time by only implementing quantum circuits of $\tcO(1)$ depth.
\end{theorem}
\begin{proof}
To obtain the claim on the sample complexity, note that for such systems $L=\Omega(1)$ by Prop.~\ref{prop:lowerbound} and they satisfy a transportation cost inequality by Thm.~\ref{thm:modifiedmlsi}. Moreover, we can learn the expectation of all $E_i$ up to an error $\epsilon$ and failure probability $\delta$ with $\cO(\log(n\delta^{-1})\epsilon^{-2})$ samples using shadows, as they all have constant support.  The claimed sample complexity then follows from Thm.~\ref{thm:quick_dirty}. The classical postprocessing result follows from~\cite{PhysRevLett.124.220601}. The postprocessing with a quantum computer follows from the results of~\cite{kastoryano2016quantum},~\cite{capel2020modified} combined with the fact that $L^{-1}U=\cO(1)$ by also invoking Lemma~\ref{lem:covariancematrix_decay}. Indeed,~\cite{kastoryano2016quantum,capel2020modified} asserts that we can approximately prepare any $\sigma(\mu)$ on a quantum computer with a circuit of depth $\tcO(1)$. Moreover, once again resorting to shadows we can estimate $e(\mu)$ with $\tcO(\epsilon^{-2})$ samples. We conclude that we can run each iteration of gradient descent in $\tcO(n\epsilon^{-2})$ time. 
As $L^{-1}U=\cO(1)$, Thm.~\ref{thm:comcomplexitygeneral} asserts that we converge after $\tcO(1)$ iterations, which yields the claim
\end{proof}
The theorem above nicely showcases how the joint use of transportation cost inequalities and strong convexity bounds to improve maximum entropy methods come together. Moreover, with access to a quantum computer, up to polylogarithmic overhead in system size factors, the computational complexity of learning the Gibbs state is comparable to reading out the results of measurements on the system. Together with the polylog sample complexity bounds that we obtain, this justifies our claiming that the result above almost gives the last word on learning such states.

\section{Comparison of sample complexity of previous methods}\label{sec:lowerbound}
Our work arguably introduces two technical innovations to the literature of learning and tomography of quantum states that underly our exponential speedups in sample complexity. The first is the observation that most observables of physical interest have a small Lipschitz constant, and, thus, it might be more motivated to look for good approximations in Wasserstein distance instead of trace distance. The second is that for states that satisfy a TC inequality, it suffices to obtain an estimate of the state that has a small relative entropy density with the target state to recover Lipschitz observables. And that finding such an estimate can be achieved from a few samples through a combination of maximum entropy and classical shadow methods.

We will now argue that these two innovations are indeed crucial to ensure our exponential speedups. First, we will show that the shadow protocol will yield bad estimates for the expectation value of local observables with high probability even for product states if the number of samples is not exponential in the locality of the underlying observables. This shows that exploiting the locality of the underlying states is crucial to obtaining a polynomial sample complexity in the locality of the observables. After that, we will show in Subsec.~\ref{sec:trace_dist_lower} that $\Omega(\sqrt{n}\epsilon^{-2})$ samples are necessary for any algorithm that can recover any Gibbs state on a regular lattice at constant temperature up to trace distance $\epsilon$. This will follow from the results of~\cite{Devroye2020} and showcase the need to focus on the Wasserstein distance instead of the trace distance.

\subsection{Lower bounds for sample complexity of classical shadows}\label{sec:failure_shadows}
One of the main advantages of our results compared with the classical shadows method of~\cite{Huang2020} is that whenever a TC inequality is available, we can learn all $k-$local observables with a number of samples that grows polynomially in $k$, whereas classical shadows require a number of samples that grows exponentially in $k$. However, the classical shadows framework does not assume any structure for the underlying state. Thus, it is natural to ask if this undesired exponential scaling of the shadows framework is due to this broader applicability. 

In this section, we will demonstrate that this is not the case even for one of the simplest imaginable classes of states, namely tensor products of Pauli eigenstates. We will show that if the number of samples is not exponential in the locality of the desired observables, there will always be a $k$-local observable whose estimate will be wrong with constant probability.

But before we show that, let us briefly recall how the shadow method works. To recover local observables on $n$ qubits, the methods of~\cite{Huang2020} proceed as follows. First, we sample a random unitary $U=\otimes_{i=0}^{n-1}U_i$, where each $U_i$ is an independent rotation into a Pauli basis. Then we proceed to measure the state of interest $\rho$ in the basis defined by $U$, obtaining a $n$-bit classical string $b_0b_1\ldots b_{n-1}$. The shadow corresponding to this measurement $\hat{\rho}$ is then defined to be given by
\begin{align*}
\hat{\rho}=\otimes_{i=1}^n (3U_i^\dagger \ketbra{b_i}U_i-\one).
\end{align*}
We then repeat this procedure $SM$ times for $S,M\in\N$, obtaining shadows $\{\hat{\rho}_{s,m}\}_{1\leq m\leq M,1\leq s\leq S}$. We then set our estimate of the expectation value of an observable $O$ to be given by
\begin{align}\label{equ:estimator_shadow}
\hat{O}(\{\hat{\rho}_{s,m}\})=\operatorname{median}\left\{S^{-1}\sum\limits_{s=1}^S\tr{\hat{\rho}_{s,m=1}O},S^{-1}\sum\limits_{s=1}^S\tr{\hat{\rho}_{s,m=2}O},\ldots,S^{-1}\sum\limits_{s=1}^S\tr{\hat{\rho}_{s,m=M}O}\right\}.
\end{align}
That is, we partition the set of shadow samples into $M$ subgroups, take the average on each of them and then take their median. The authors of~\cite{Huang2020} then proceed to show in Theorem 1 that if we take $SM=\cO(4^k\log(N\delta^{-1})\epsilon^{-2})$, then with probability at least $1-\delta$ the expectation value any given $N$ $k$-local observables of bounded operator norm will deviate by at most $\epsilon$ from the estimate given in Eq.~\eqref{equ:estimator_shadow}. We will now show that this exponential scaling in $k$ is unavoidable if we want to obtain nontrivial estimates with high probability. More precisely:
\begin{prop}
Let $\rho=\otimes_{i=1}^n\ketbra{\phi_i}$ be an unknown $n$ qubit state where each $\ket{\phi_i}$ is given by a Pauli eigenstate. For $k\geq \log_3(n)\log(n)$, if $SM\leq n$ then there is an observable $O$ of the form
\begin{align*}
O=n^{-1}\sum\limits_{i=1}^nO_i
\end{align*}
where each $O_i$ is $k$-local and such that:
\begin{align}
\tr{O\rho}=1,\quad \hat{O}(\{\hat{\rho}_{s,m}\})=0
\end{align}
with probability at least $1-e^{-1}$.
\end{prop}

\begin{proof}
We will let $O_i=\otimes_{j=i}^{i+k}P_j$, where we take addition modulo $n$ and let $P_i$ be the Pauli matrix such that $P_i\ket{\phi_i}=1$, where we assume w.l.o.g. that each Pauli eigenstate corresponds to a $+1$ eigenstate.  It is then clear that $\tr{O\rho}=1$.

Let us now analyse the performance of shadows. Note that as the random unitaries used to obtain each sample are random Pauli bases, we have that $\tr{O_i\hat{\rho}_{s,m}}=3^k$ if the unitary $U^{s,m}$  corresponds to the identity on qubits $i,i+1,\ldots,i+k$ and $0$ otherwise. This is because then the string we measure will be rotated to a Pauli basis different from that of $P_i$ on at least one of the qubits in that interval. Thus, as we pick the Pauli bases uniformly at random and there are three different possibilities for each one of them, we see that the probability that a shadow is different from $0$ on a given $O_i$ is $3^{-k}$. By a union bound, the probability that a shadow is different than $0$ on any of the $n$ $O_i$ is at most $n3^{-k}$. As the different shadows are independent, the probability that all of the $SM$ shadows return $0$ is at least $(1-n3^{-k})^{SM}$. As we picked $k\geq \log_3(n)\log(n)$ and $SM\leq n$, we have:
\begin{align*}
(1-n3^{-k})^{SM}\geq (1-{n}^{-1})^{n}\ge \frac{e^{-1}}{2}
\end{align*}
for $n$ large enough. Thus, as all shadows will have expectation value $0$, the median and means procedure will clearly also output $0$, which concludes the proof.
\end{proof}
Despite the fact the proof above used quite rough estimates and simple observables and states, it still gives some intuition as to why classical shadows require an exponential number of samples in the locality. We see that the probability that the shadows "look in the right direction" is exponentially small in the locality of the observables, in the sense that the overlap between the state and a random Pauli eigenstate is likely to be exponentially small. And whenever it does look in a direction in which the underlying state has a significant overlap with that basis it has to compensate that direction exponentially. Thus, if the number of samples is not exponential in the locality, it is unlikely that we will measure in a direction that has significant overlap with our state or there will be significant fluctuations due to the exponential rewarding of the "good" directions.

However, by combining shadows with a locality structure and the maximum entropy principle, as we do in this work, we can bypass the need to measure in random directions for observables with a large locality, bypassing this exponential scaling.

We also note that the authors of~\cite{Huang2020} already proved the optimality of their protocol by only considering product states in Section 8 of their supplemental material. The main difference between their proof and ours is that we focus on a Lipschitz observable, whereas they focus on observables that only depend on $k$ qubits.

\subsection{Lower bounds for recovery in trace distance}\label{sec:trace_dist_lower}
In the main text, we claimed that one of the reasons why we obtain an exponential speedup compared to usual many-body methods is that we focus on a good recovery in the Wasserstein distance instead of trace distance. Moreover, by combining the maximum entropy method with a TC inequality we are able to obtain good recovery guarantees from a constant \emph{relative entropy density}. That is, as long as two Gibbs states $\sigma(\lambda),\sigma(\mu)$ satisfy
\begin{align}\label{equ:relative_density}
D(\sigma(\lambda)\|\sigma(\mu))\leq \epsilon n,
\end{align}
for some $\epsilon>0$ we already obtain some nontrivial guarantees.

In this section, we will argue that focusing on the Wasserstein recovery instead of trace distance is essential to obtain nontrivial recovery guarantees from a number of samples scaling logarithmically with the system size. To achieve this, we will resort to results of~\cite[Theorem 1.3]{Devroye2020}:
\begin{prop}\label{prop:lower_trace}[Lower bound of sample complexity in trace distance]
Let $G=(V,E)$ be a graph on $n$ vertices and $m$ edges and for $\lambda\in\R^{15m}$, $\|\lambda\|_{\ell_\infty}\leq 1$ let $H(\lambda)$ be defined as
\begin{align*}
H(\lambda)=\sum\limits_{i\sim j}\sum_{l=1}^{15} \lambda_{i,j}^lH_{i,j}^l,
\end{align*}
where the $H_{i,j}^l$ correspond to some ordering of the nonidentity Pauli strings acting on sites $i,j$. Then for any  $\beta=\Omega(m^{-\frac{1}{2}})$, let $\hat{\sigma}(\lambda)$ be the estimate of $\sigma(\lambda)$ outputted by an algorithm with access to $s$ samples from a state $\sigma(\lambda)$. Then:
\begin{align}
\sup\limits_{\lambda\in B_\infty(0,1)}\|\sigma(\lambda)-\hat{\sigma}(\lambda)\|_{\operatorname{tr}}=\Omega\left(\min\left\{1,\sqrt{\frac{m}{s}}\right\}\right).
\end{align}
\end{prop}

\begin{proof}
This statement immediately follows from~\cite[Theorem 1.3]{Devroye2020}, which shows the analogous statement when restricted to classical Ising models. As our class of Hamiltonians includes those as a subset, any algorithm that could provide an estimate for this more general class also can find one for the classical instances. Moreover, in the proof of \cite[Theorem 1.3]{Devroye2020} the inverse temperature is absorbed into the coefficients of the Hamiltonian, which are assumed to have $2$-norm bounded by a constant independent of the system's size. This is easily seen to be satisfied by our conditions since $\beta=\Omega(m^{-\frac{1}{2}})$.
\end{proof}
The statement above implies in particular that any algorithm that finds an estimate that is $\epsilon$ close in trace distance for all $2$-local Gibbs state on a lattice and constant inverse temperature requires $\Omega(n\epsilon^{-2})$ samples. In contrast, we see that it is possible to obtain an estimate that is $\cO(\epsilon\sqrt{n})$ close in Wasserstein distance from $\cO(\epsilon^{-2}\log(n))$ samples of the Gibbs states, which is sufficient to already give nontrivial recovery guarantees for Lipschitz observables. Thus, we see from Prop.~\ref{prop:lower_trace} that resorting to the Wasserstein distance is essential to obtain recovery guarantees in the regime where the number of samples is logarithmic in the system's size.

Furthermore, it is interesting to note that the proof of~\cite{Devroye2020} is based on a set of Gibbs states of the form:
\begin{align}\label{equ:lower_Gibbs}
\delta\sum\limits_{i\sim j}s_{i,j}Z_iZ_j,
\end{align}
with $\delta=\Theta(m^{-\frac{1}{2}})$ and $s_{i,j}\in\{\pm1\}$. Their proof then proceeds by finding a large subset of Gibbs states of the form in Eq.~\eqref{equ:lower_Gibbs} which have a trace distance and relative entropy of constant order. The lower bound on the sample complexity then follows from standard information-theoretic arguments. We believe that this class of examples in the proof further illustrates why the trace distance is not necessarily the adequate distance measure when estimating the error on extensive observables. Indeed, for extensive, local observables the class of Gibbs states from the Hamiltonians in Eq.~\eqref{equ:lower_Gibbs} behaves like the maximally mixed state, as each local term converges to $0$ as the system size increases.

\end{document}